\documentclass[11pt]{article}

\usepackage{amsmath,amssymb,amsthm,verbatim,relsize,mathrsfs,hyperref}
\usepackage{cleveref}
\hypersetup{colorlinks = true,linkcolor = blue, anchorcolor = blue, citecolor = blue, filecolor = red,urlcolor = red}
\usepackage{pgf}
\usepackage{tikz}
\usepackage{algorithm}
\usepackage[noend]{algpseudocode}
\usetikzlibrary{arrows,automata}
\usepackage{tipa}
\usepackage [english]{babel}
\usepackage [autostyle, english = american]{csquotes}
\usepackage{bbold}
\usepackage{subfiles}
\usepackage{blindtext}
\usepackage{framed}
\usepackage{mdframed}
\usepackage{fullpage}
\usepackage[nottoc]{tocbibind}
\usepackage{paralist}
\usepackage{enumitem}
\MakeOuterQuote{"}

\makeatletter
\newcommand\footnoteref[1]{\protected@xdef\@thefnmark{\ref{#1}}\@footnotemark}
\makeatother

\newtheorem{theorem}{Theorem}[section]
\newtheorem{definition}[theorem]{Definition}
\newtheorem{lemma}[theorem]{Lemma}

\newtheorem{remark}[theorem]{Remark}
\newtheorem{observation}[theorem]{Observation}
\newtheorem{fact}[theorem]{Fact}
\newtheorem{claim}[theorem]{Claim}

\newcommand{\Expect}{{\rm I\kern-.3em E}}

\newcommand{\ignore}[1]{}

\newcommand{\cD}{\mathcal{D}}

\newcommand{\cH}{{\cal H}}

\newcommand{\cM}{{\cal M}}
\newcommand{\cP}{\mathcal{P}}

\newcommand{\cU}{{\cal U}}

\newcommand{\eps}{\varepsilon}

\newcommand{\poly}{\mathrm{poly}}
\newcommand{\polylog}{\mathrm{polylog\xspace }}

\newcommand{\calE}{{\cal E}}

\newcommand{\calP}{{\cal P}}

\newcommand{\bone}{\mathbf{1}}

\newcommand{\be}{\mathbf{e}}
\newcommand{\bp}{\boldsymbol{p}}

\newcommand{\bA}{\mathbf{A}}

\newcommand{\Exp}{\EX}

\newcommand{\EX}{\hbox{\bf E}}

\newcommand{\Sec}[1]{\hyperref[sec:#1]{\S\ref*{sec:#1}}} 
\newcommand{\Eqn}[1]{\hyperref[eq:#1]{(\ref*{eq:#1})}} 
\newcommand{\Fig}[1]{\hyperref[fig:#1]{Fig.\,\ref*{fig:#1}}} 
\newcommand{\Tab}[1]{\hyperref[tab:#1]{Tab.\,\ref*{tab:#1}}} 
\newcommand{\Thm}[1]{\hyperref[thm:#1]{Theorem\,\ref*{thm:#1}}} 
\newcommand{\Fact}[1]{\hyperref[fact:#1]{Fact\,\ref*{fact:#1}}} 
\newcommand{\Lem}[1]{\hyperref[lem:#1]{Lemma\,\ref*{lem:#1}}} 
\newcommand{\Prop}[1]{\hyperref[prop:#1]{Prop.~\ref*{prop:#1}}} 
\newcommand{\Cor}[1]{\hyperref[cor:#1]{Corollary~\ref*{cor:#1}}} 
\newcommand{\Conj}[1]{\hyperref[conj:#1]{Conjecture~\ref*{conj:#1}}} 
\newcommand{\Def}[1]{\hyperref[def:#1]{Definition~\ref*{def:#1}}} 
\newcommand{\Alg}[1]{\hyperref[alg:#1]{Alg.~\ref*{alg:#1}}} 
\newcommand{\Ex}[1]{\hyperref[ex:#1]{Ex.~\ref*{ex:#1}}} 
\newcommand{\Clm}[1]{\hyperref[clm:#1]{Claim~\ref*{clm:#1}}} 
\newcommand{\Step}[1]{\hyperref[step:#1]{Step~\ref*{step:#1}}} 
\newcommand{\Obs}[1]{\hyperref[obs:#1]{Observation~\ref*{obs:#1}}} 

\newcommand{\AH}{\mathbf{A}_{n,d}}
\newcommand{\AHl}{\mathbf{A}_{n,1}}

\begin{document}

\title{A $o(d) \cdot \polylog~n$ Monotonicity Tester for Boolean Functions \\ over the Hypergrid $[n]^d$}

\author{
	Hadley Black\thanks{Department of Computer Science, University of California, Santa Cruz. Email: \href{mailto:hablack@ucsc.edu}{\nolinkurl{hablack@ucsc.edu}}.}
	\and
	Deeparnab Chakrabarty\thanks{Department of Computer Science, Dartmouth College. Email: \href{mailto:deeparnab@dartmouth.edu}{\nolinkurl{deeparnab@dartmouth.edu}}.}
	\and 
	C. Seshadhri\thanks{Department of Computer Science, University of California, Santa Cruz. Email: \href{mailto:sesh@ucsc.edu}{\nolinkurl{sesh@ucsc.edu}}.}
}

\maketitle
\begin{abstract}
We study monotonicity testing of Boolean functions over the hypergrid $[n]^d$ and design a non-adaptive tester with $1$-sided error whose query complexity is $\tilde{O}(d^{5/6})\cdot \poly(\log n,1/\eps)$. Previous to our work, the best known testers had query complexity linear in $d$ but independent of $n$. We improve upon these testers as long as $n = 2^{d^{o(1)}}$. 

To obtain our results, we work with what we call the {\em augmented hypergrid}, which adds extra edges to the hypergrid. 
Our main technical contribution is a Margulis-style isoperimetric result for the augmented hypergrid, and
our tester, like previous testers for the hypercube domain, performs directed random
walks on this structure.
\end{abstract}
\thispagestyle{empty}
\clearpage
\setcounter{page}{1}
\newpage
\section{Introduction}
Monotonicity testing is a classic property testing problem that asks whether a function defined over a partial order is monotone or not. 
Consider a function $f: D \to R$ (where $D$ is a partial
order and $R$ is an ordered range). The function
$f$ is monotone if $f(x) \leq f(y)$ whenever $x < y$
in the partial order $D$. The \emph{distance}
between two functions $f$ and $g$ is the fraction of points they
differ in. The distance to monotonicity of $f$ is $\min_{g \in \cP}
d(f,g)$, where $\cP$ is the set of monotone functions.
Given a parameter $\eps \in (0,1)$,
the aim of a property tester is to correctly determine, with high probability, whether $f$
is monotone or the distance to monotonicity is at least $\eps$. When the distance to monotonicity of $f$ is at least $\eps$, we say that $f$ is $\eps$-far from being monotone.

In recent years, there has been a lot of work~\cite{GGLRS00,ChSe13-j,ChenST14,ChenDST15,KMS15,BeBl16,Chen17} on understanding the testing question for Boolean functions defined over the $d$-dimensional hypercube $\{0,1\}^d$ domain. This line of work has unearthed a connection between monotonicity testing and isoperimetric theorems on the {\em directed} hypercube. 

In this paper, we investigate monotonicity testing of Boolean functions over the $d$-dimensional $n$-hypergrid, $[n]^d$. 
Apart from being a natural property testing question, our motivation
is to unearth isoperimetry theorems for richer structures.
%
Indeed, our main technical contribution is
a Margulis-type isoperimetry theorem for 
a structure called 
the augmented hypergrid. Such a theorem allows us to design a tester with query complexity $o(d) \cdot \polylog~n$ for Boolean functions defined on $[n]^d$. As long as $n = 2^{d^{o(1)}}$, this $o(d)$-query tester has the best query complexity among the testers known so far.

\begin{theorem}\label{thm:ourmainresult}
	Given a function $f:[n]^d\to \{0,1\}$ and a parameter $\eps$, there is a randomized algorithm that makes $O(d^{5/6}\cdot \log^{3/2}d\cdot (\log n + \log d)^{4/3} \cdot \eps^{-4/3})$ non-adaptive queries and (a) returns YES with probability $1$ if the function is monotone, and (b) returns NO with probability $>2/3$ if the function is $\eps$-far from being monotone.
\end{theorem}

\subsection{Perspective}
\paragraph{The Hypercube.}
Goldreich et al.~\cite{GGLRS00} gives the first monotonicity testers for Boolean functions over the $\{0,1\}^d$ hypercube. Their 
tester makes $O(d/\eps)$ queries. Chakrabarty and Seshadhri~\cite{ChSe13-j} describes a $o(d)$ query tester via so-called directed isoperimetry theorems.
Given $f$, define $S^{-}_f$ to be the set of edges of the hypercube with their "lower endpoint" evaluating to $1$ and their "upper endpoint" evaluating to $0$. That is, $S^{-}_f$ is the set of edges violating monotonicity. Define $I^{-}_f := |S^{-}_f|/2^d$ and $\Gamma^{-}_f := |\textrm{max-match}(S^{-}_f)|/2^d$ where the numerator indicates cardinality of the largest matching in $S^{-}_f$.~\cite{ChSe13-j} prove that if $f$ is $\eps$-far from being monotone, then $I^{-}_f\cdot \Gamma^{-}_f = \Omega(\eps^2)$. 
If the edges of the hypercube were to be oriented from the point with fewer ones to the one with more, then this result connects the structure of directed edges leaving the set of points evaluating to one with the distance to monotonicity.
In this sense, this result is a directed analogue of a result by Margulis~\cite{Mar74} which proves a similar statement on the undirected hypercube.
Using this directed isoperimetry result,~\cite{ChSe13-j} designs an $O(d^{7/8}\eps^{-3/2})$-query tester. 
Chen et al.~\cite{ChenST14} refines the analysis in~\cite{ChSe13-j}
to give a tester with $O(d^{5/6}\eps^{-4})$ query complexity.\smallskip


A remarkable paper of Khot, Minzer, and Safra~\cite{KMS15} proves the following directed analogue of Talagrand's~\cite{Tal93} isoperimetry theorem.
If $S^{-}_f(x)$ is the number of edges in $S^{-}_f$ incident on $x$, then~\cite{KMS15} proves $\Exp_x [\sqrt{|S^{-}_f(x)|}] = \Omega(\eps/\log d)$. 
The directed analogue of Talagrand's~\cite{Tal93} isoperimetry theorem is 
stronger than the Margulis-type theorem (albeit with an extra $\log d$ in the denominator.). 
Khot et al.~\cite{KMS15} uses this stronger directed isoperimetry
result to obtain a $\tilde{O}(\sqrt{d}\eps^{-2})$-query monotonicity tester.
This bound is nearly optimal for 
non-adaptive testers~\cite{ChenDST15,KMS15}. 

However, the proof techniques of both these isoperimetry
results are very different. 
Chakrabarty-Seshadhri~\cite{ChSe13-j} use the combinatorial structure of the "violation graph" to explicitly find either a large number of edges in $S^{-}_f$, or to find a large matching in $S^{-}_f$. Khot et al.~\cite{KMS15} instead propose an operator (the split operator) which converts a function that is far from monotone to a function with sufficient structure that allows them to prove the Talagrand-type isoperimetry theorem in a relatively easier way. This technique of~\cite{KMS15} is reminiscent of the original result of Goldreich et al.~\cite{GGLRS00} which also defines an operator (the switch operator) to convert a function to a monotone function and accounting for the number of violated edges. 
It appears that methods which change function values are harder to generalize for the hypergrid domain. In particular, it is not clear how to generalize the switch or the split operators for hypergrids.

\paragraph{The Hypergrid.}
Dodis et al.~\cite{DGLRRS99} is the first paper to study property testing on the $d$-dimensional hypergrid $[n]^d$. For Boolean functions, this paper describe a 
$O(\frac{d}{\eps}\log^2(\frac{d}{\eps}))$-query tester. Note that the query complexity is independent of $n$.
The proof follows via  {\em dimension reduction} theorem for Boolean functions.
 This result asserts that
  if a Boolean function on the $[n]^d$ hypergrid is $\eps$-far from being monotone, then the function restricted to a random line has an expected distance of $\Omega(\eps/d)$ to monotonicity. On a line it is not too hard to see that Boolean functions can be tested with $\tilde{O}(1/\eps)$-queries. This style of analysis was refined by Berman, Raskhodnikova and Yarovslavstev~\cite{BeRaYa14} which gives a $O(\frac{d}{\eps}\log(\frac{d}{\eps}))$-query tester for Boolean functions on hypergrids. This is the current best known tester.
Since these testers project to a line they (a) have no dependence on $n$, and (b) they seem to need the linear dependence on $d$ since the violations may be restricted to $O(1)$ unknown dimensions which, if naively done, may take $\Theta(d)$ queries to detect.

For {\em real-valued} functions over $[n]^d$, Dodis et al.~\cite{DGLRRS99} give a $O(\frac{d}{\eps}\log n\log |\mathrm{R}|)$-query tester where $\mathrm{R}$ is the range of the function. 
They do so by a clever range-reduction technique that
reduces to testing Boolean functions over $[n]^d$. 
One of the key ideas to emerge
from results of Ergun et al.~\cite{EKK+00} and Bhattacharyya et al.~\cite{BGJ+12} on monotonicity testing on the line
(and richer structures) is to compare points that
are far apart.
%
%
Chakrabarty and Seshadhri~\cite{ChSe13} exploit this idea
to give an optimal $O(\frac{d}{\eps}\log n)$-query tester for real-valued functions over $[n]^d$, removing the dependence on $\mathrm{R}$.
Specifically, their tester queries pairs in the hypergrid that
may be apart by an arbitrary power of $2$.
One can think of adding these extra edges to 
get an \emph{augmented hypergrid}.
(This is the central theme
of the transitive closure spanner idea of Bhattacharyya et al~\cite{BGJ+12}.)
This notion of the {\em augmented hypergrid} is central to our paper. The main result of~\cite{ChSe13} was to show that if (even a real-valued) $f$ is $\eps$-far from being monotone, then this augmented hypergrid has many violated edges. For $\eps$-far Boolean valued functions, this implies that the "out-edge-boundary" of the set of $1$s must be large. 

The main technical result of this paper is proving a Margulis-style result for the augmented hypergrid
generalizing the result of~\cite{ChSe13-j}. It states that either the "out-edge-boundary" is "very large", or the "out-vertex-boundary" is large 
(details in \Sec{isoperimetry}.).
One of the main tools that~\cite{ChSe13-j} use is a routing theorem in the hypercube due to Lehman and Ron~\cite{LR01}. One of the ways this theorem is proved and used exploits the fact that the "directed hypercube" is a layered DAG with vertices of the same Hamming weight forming the layers. The "directed hypergrid" is also a layered DAG, but the augmented hypergrid is not. This technically poses many challenges, and our way out is to define "good portions" of the hypergrid where a certain specified subgraph is indeed layered. We generalize Lehman-Ron, but more crucially we can show if a function is $\eps$-far, then large good portions exists. The definitions of these good portions is perhaps our main conceptual combinatorial contribution.

\subsection{Reducing to the case when $n$ is a power of $2$}\label{sec:discussion_power}

It greatly simplifies the presentation to assume that $n$ is a power of $2$. For monotonicity testing, this is no
loss of generality. In \Sec{power}, we show that monotonicity testing over general hypergrids can be reduced to the case when $n$ is a power of $2$. 
Specifically, in \Thm{power} we reduce testing over general $[n]^d$ to testing over $[N]^d$ where $N$ is a power of 2 and $N = \Theta(nd)$. 
In our case, this incurs a loss of $\polylog~d$ in the query complexity. 
Thus, we assume that $n$ is a power of $2$ throughout the paper except in \Thm{ourmainresult}, where, in the query complexity, $\log n$ is replaced by $\log n + \log d$ to reflect this loss. To be specific, \Sec{good} and \Sec{lehmanron} \emph{do not need} $n$ to be a power of $2$, while we stress that \Sec{tester_analysis}, \Sec{fourier} and \Sec{edge} \emph{do need} $n$ to be a power of $2$. 

\subsection{The Augmented Hypergrid}\label{sec:augmented_hypergrid}

Given the $d$-dimensional $n$-hypergrid $[n]^d$, we define the {\em augmented hypergrid} $\AH$ which is simply the standard hypergrid with additional edges connecting any two vertices which differ in exactly one dimension by a power of two in the range $1 \leq 2^a \leq n$. This construction was explicitly introduced in \cite{ChSe13}.

It is useful to partition the edges of $\AH$ into a collection of matchings $\mathbf{H} := \{H^c_{i,a} : i \in [d], a \in [\log n], c\in \{0,1\}\}$, where
\begin{itemize}[noitemsep]
	\item $H_{i,a}^0 := \{(x,y) : y_i - x_i = 2^a, x_j = y_j \text{ } \forall j \neq i, x_i \pmod{2^{a+1}} < 2^a \}$.
	\item $H_{i,a}^1 := \{(x,y) : y_i - x_i = 2^a, x_j = y_j \text{ } \forall j \neq i, x_i \pmod{2^{a+1}} \geq 2^a \}$.
\end{itemize}
Note that $H^0_{i,a}$ is a perfect matching, but $H^{1}_{i,a}$ is not.
We let $d_{\bA}(x,y)$ denote the shortest-path distance between two points in the augmented hypergrid.

\begin{figure}
  \hspace*{5cm} \includegraphics[width=.45\linewidth]{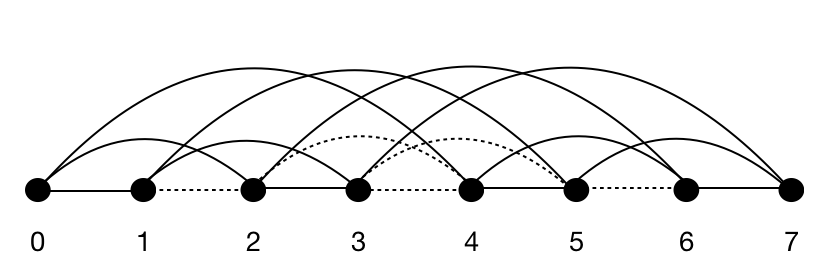}
  \caption{\small{The augmented line $\mathbf{A}_{8,1}$. Solid and dashed lines represent edges from the "even" and "odd" matchings, respectively. E.g. the solid lines along the bottom make up $H^0_{1,0}$ ("even" matching of length $1 = 2^0$ edges) and the dashed lines along the bottom make up $H^1_{1,0}$ ("odd" matching of length $1 = 2^0$ edges). } }
  \label{fig:aug_line}
\end{figure}

\subsection{The Monotonicity Tester}

Our tester is a generalization of the tester described by Khot, Minzer, and Safra~\cite{KMS15} over the Boolean hypercube, 
which itself is inspired by the path tester described in~\cite{ChenST14,ChSe13-j}. Instead of taking a random walk on the hypergrid, however, we perform a random walk on the augmented hypergrid.\smallskip

\begin{figure}[h]
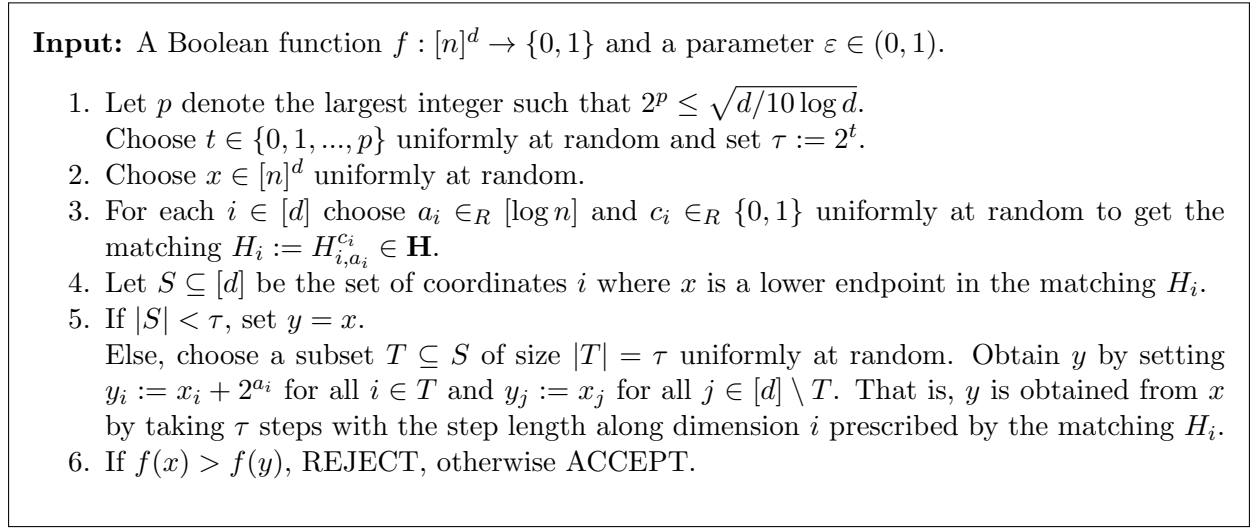

\begin{framed}
    \noindent \textbf{Input:} A Boolean function $f: [n]^d \to \{0,1\}$ and a parameter $\eps \in (0,1)$.
	\begin{enumerate}[noitemsep]
		\item \label{step:tau} Let $p$ denote the largest integer such that $2^p \leq \sqrt{d/10\log d}$.  \\
			  Choose $t \in \{0,1,...,p\}$ uniformly at random and set $\tau := 2^t$.
		\item Choose $x \in [n]^d$ uniformly at random.   
		\item \label{step:H} For each $i \in [d]$ choose $a_i \in_R [\log n]$ and $c_i \in_R \{0,1\}$ uniformly at random to get the matching $H_i := H_{i,a_i}^{c_i} \in \mathbf{H}$.
		\item \label{step:S} Let $S\subseteq [d]$ be the set of coordinates $i$ where $x$ is a lower endpoint in the matching $H_i$.
        \item \label{step:y} If $|S| < \tau$, set $y=x$. \\ 
        Else, choose a subset $T\subseteq S$ of size $|T|=\tau$ uniformly at random. 
		Obtain $y$ by setting $y_i := x_i + 2^{a_i}$ for all $i\in T$ and $y_j := x_j$ for all $j \in [d] \setminus T$. That is, $y$ is obtained from $x$ by taking $\tau$ steps with the step length along dimension $i$ prescribed by the matching $H_i$.
		\item If $f(x) > f(y)$, REJECT, otherwise ACCEPT.
	\end{enumerate} 
\end{framed}
\caption{\small{\textbf{Monotonicity Tester for Boolean Functions on} $[n]^d$}}
\label{fig:alg}
\end{figure}

\noindent
The main result of this paper is the following theorem which easily implies \Thm{ourmainresult}.
\begin{theorem}\label{thm:main-theorem}
	Given a function $f:[n]^d\to\{0,1\}$ which is $\eps$-far from being monotone, the tester described in \Fig{alg} with inputs $f$ and $\eps$ detects a violation with probability $\Omega\left(\frac{\eps^{4/3}}{d^{5/6}\log^{3/2}d \log^{4/3} n}\right)$.
\end{theorem}
\subsection{Related Work and Remarks}
Monotonicity Testing has been extensively~\cite{EKK+00,GGLRS00,DGLRRS99,LR01,FLNRRS02,HK03,AC04,HK04,ACCL04,E04,SS06,Bha08,BCG+10,FR,BBM11,RRSW11,BGJ+12,ChSe13,ChSe13-j,ChenST14,BeRaYa14,BlRY14,ChenDST15,ChDi+15,KMS15,BeBl16,Chen17} studied in the past two decades; in this section we discuss a few works most relevant to this paper.

In property testing, the notion of distance between functions is usually the Hamming distance between them, that is, the fraction of points at which they differ. More generally one can think of a general measure over the domain and the distance is the measure of the points at which the two functions differ. Monotonicity testing has been studied~\cite{AC04,HK04,ChDi+15} over general product measures. It is now known~\cite{ChDi+15} that for functions over $[n]^d$, there exist testers making $O(\frac{d}{\eps}\log n)$-queries over any product distribution; in fact there exist better testers if the distribution is known. A simple argument (Claim 3.6 in~\cite{ChDi+15}) shows that testing monotonicity of Boolean functions over $\{0,1\}^d$ over any product distribution reduces to testing over $[n]^d$ over the uniform distribution. Thus our result gives $o(d)$-query monotonicity testers for $f:\{0,1\}^d\to \{0,1\}$, even over $p$-biased distributions; this holds even when $p_i$'s are not constants and depend on $d$. Once again, it is not clear how to generalize the tester of Khot, Minzer, and Safra~\cite{KMS15} to obtain such a result.


In a different take on the distance function, 
Berman, Raskhodnikova, and Yaroslavtsev~\cite{BeRaYa14} study property testing when the distance function is not the Hamming distance but could be a more general $L_p$ norm. That is the distance between $f$ and $g$ is the $L_p$ norm of $f-g$ (the measure over which the $L_p$ norm is taken is the uniform measure). One of their results is a black-box reduction (Lemma 2.2 in~\cite{BeRaYa14}) of $L_1$-testing of functions in the range $[0,1]$ to "usual $L_0$"-testing of Boolean functions. In particular, our result along with their reduction implies $o(d) \cdot \polylog n$-query testers for $L_1$-testing over $[n]^d$. Another interesting result in the same paper is a {\em separation} between non-adaptive and adaptive testers. For Boolean functions defined over $[n]^2$, Berman et al.~\cite{BeRaYa14} describe an $O(1/\eps)$-query adaptive tester and a lower bound of $\Omega(\frac{1}{\eps}\log(\frac{1}{\eps}))$-for non-adaptive testers. It is an interesting question (even for the hypercube) whether adaptivity helps in Boolean monotonicity testing; it is known for real-valued functions it doesn't~\cite{ChSe14}. Some recent results~\cite{BeBl16,Chen17} point out some very interesting lower bounds for adaptive testers.

Monotonicity testing is well-defined over any arbitrary poset. Our knowledge here is limited. Fischer et al.~\cite{FLNRRS02} prove there exist $O(\sqrt{N/\eps})$-query testers over any poset of cardinality $N$ even for real-valued functions; they also prove an $\Omega(N^{\frac{1}{\log\log N}})$-lower bound even for Boolean functions. 
On the other hand, there are good testers for the hypercube and hypergrid even for real-valued functions. 
Can we understand the structure that allows for efficient testers?
Our notion of "good portions" (\Lem{good}) holds for {\em any} poset, and may
provide some directions towards this question.

Finally, we comment on our tester's dependence on $n$. 
If does not seem possible to improve our current line of attack,
since the number of edges in the augmented hypergrid (when divided by $n^d$) depends on $n$. One direction may be to sparsify the augmented hypergrid in such a way that we don't lose out on the Margulis-type inequality. 
It is an interesting direction to get a greater understanding of such isoperimetric inequalities
and possibly remove this dependence on $n$.

\section{Isoperimetric Theorems on the Augmented Hypergrid}\label{sec:isoperimetry}

Given a function $f:[n]^d\to\{0,1\}$ we consider it to be defined over the vertices of $\AH$. 
We let $S^{-}_f$ denote the set of edges $(x,y)$ of $\AH$ with $f(x) = 1$ and $f(y) = 0$, where $x$ is the lower endpoint. We let $I^-_f := |S^-_f|/n^d$. If $f:\AH\to \{0,1\}$ is $\eps$-far from being monotone, then $I^-_f = \Omega(\eps)$. This result is
implicit in many earlier papers~\cite{DGLRRS99,EKK+00,ChSe13} on monotonicity testing over the hypergrid.

If one considers the edges of $\AH$ being oriented from the lower to the upper endpoint, then the above theorem lower bounds the normalized "out-edge-boundary" of the indicator set of a function which is far from monotone. It is instructive to note that to obtain this result one needs to look at the augmented hypergrid. If one considered the standard hypergrid then one would need an extra $n$-factor in the denominator in the RHS. This is apparent even when $d=1$ and the function is $1$ on the first half of the line and $0$ on the second half.

One can also think about the normalized out-vertex-boundary of $f$ in $\AH$ defined as 
$\frac{1}{n^d}\cdot\left|\{x: \exists y, (x,y)\in S^-_f\}\right|$. In fact, we focus on the following smaller quantity
\[
\Gamma^{-}_f := \frac{1}{n^d}\cdot(\textrm{Size of the maximum-cardinality matching in $S^-_f$})
\] 
Our main technical result is the following Margulis-style~\cite{Mar74} directed isoperimetry theorem over the augmented hypergrid.
\begin{theorem}\label{thm:main-margulis}
	If $f:\AH\to \{0,1\}$ is $\eps$-far from being monotone, then 
	\[
	I^{-}_f\cdot \Gamma^{-}_f = \Omega(\eps^2)
	\]
\end{theorem}
As in the Boolean hypercube case, the Margulis-style isoperimetry theorem allows one to analyze the tester in \Fig{alg}. 
We follow the clean analysis of Khot et al~\cite{KMS15} and discuss this in \Sec{tester_analysis}. 
\smallskip

The proof of \Thm{main-margulis} follows the structure of that in Chakrabarty and Seshadhri~\cite{ChSe13-j} for functions over the Boolean hypercube. Fix a function $f$ which is $\eps$-far from being monotone. Consider a matching on the vertices of $\AH$ consisting of disjoint pairs of violations (a matching in the violation graph of $f$). A folklore theorem states that any such matching that is \textit{maximal} has cardinality $\geq \eps n^d/2$. 
We focus on maximal matchings in the violation graph of $f$ which {\em minimize} the average shortest path distance in $\AH$ between its endpoints. That is, we consider $M$ which minimizes $\frac{1}{|M|}\sum_{(x,y)\in M} d_\bA(x,y)$.
Let $r$ be this minimum value.

The following theorem can be proved using the techniques developed in~\cite{ChSe13-j,ChSe13}. 
In \cite{ChSe13-j}, Chakrabarty and Seshadhri prove that $I_f^- = \Omega(r\eps)$ for Boolean functions over the hypercube $\{0,1\}^d$. A similar observation holds for real-valued functions over $\AH$. Indeed, \cite{ChSe13} proves $I_f^- = \Omega(\eps)$ for such functions. We show that $I_f^- = \Omega(r\eps)$ holds for Boolean functions over $\AH$ and defer the complete proof to the appendix \Sec{edge} as it is not the main contribution of this paper and the techniques are very similar to the previous works discussed.

\begin{theorem}\label{thm:edge}
	If the average distance between the endpoints of $M^*$ is $r$, then $I^{-}_f = \Omega(r\eps)$.
\end{theorem}
Therefore, if $r$ is large then the edge boundary is large. \Thm{main-margulis} follows from the next theorem which shows if $r$ is small, then there is a large matching in $S^{-}_f$. 
\begin{theorem}\label{thm:vertex}
	If the average distance between the endpoints of $M^*$ is $r$, then $\Gamma^{-}_f = \Omega(\eps/r)$.
\end{theorem}
The above theorem is where the novelty of this paper lies. In the next subsection, we make key definitions and outline the roadmap of the proof of \Thm{vertex} 
which constitutes the bulk of the paper.
\Sec{tester_analysis} contains the final analysis of the tester which follows from the above theorem via by now standard analysis procedures. The interested reader can read \Sec{tester_analysis} independent of the remainder of the paper.

\subsection{Proof of \Thm{vertex}: A Roadmap}
Among all maximal matchings which have average distance $r$, choose $M^*$ to be the one which {\em maximizes} the following potential function
\[
\Psi(M) := \sum_{(x,y) \in M} d^2_\bA(x,y)
\]
Maximizing $\Psi(\cdot)$ has the effect of uncrossing pairs in the matching, which is useful when we use $M^*$ for finding structured subgraphs in the augmented hypergrid. We also point out that this is the same potential function used in \cite{ChSe13-j} for the hypercube case. Let $M^*_i \subseteq M^*$ be the pairs $(x,y)$ with $d_\bA(x,y)=i$.
Since the average distance of $M^*$ is $r$, we get $\sum_{i\leq 2r} |M^*_i| \geq |M|/2 \geq \eps n^d/4$. 
For any $i$, let $S^*_i$ be the "lower endpoints" of $M^*_i$ which evaluate to $1$ and $T^*_i$ be the "upper endpoints" which evaluate to $0$. 
We now make a few definitions.

\begin{definition} [Consistent Sets] \label{def:consistent}
	A pair of subsets $(S,T)$ of any poset $G$ is said to be {\em $\ell$-consistent} if there exists a bijection $\phi:S\to T$ such that 
	$d_G(s,\phi(s)) = \ell$.
\end{definition}
\noindent
Note that for all $i$, we have that the sets $(S^*_i,T^*_i)$ are $i$-consistent.
The following definitions are key for proving the theorem.
\begin{definition}[Cover Graph induced by Consistent Sets]\label{def:cover}
	Given a pair $(S,T)$ of $\ell$-consistent sets in a poset $G$, we let $\calP^{(S,T)}$ denote the collection of paths in $G$ which originate from some vertex $s\in S$, terminate in some vertex $t\in T$, is a shortest path from $s$ to $t$, and has length {\em exactly} $\ell$. 
	The $\ell$-cover graph $G^{(S,T)}$ is a subgraph of $G$ formed by taking the union of all paths in $\calP^{(S,T)}$.
\end{definition}
We remark here that in any $\ell$-consistent pair $(S,T)$ there may be $s\in S$ and $t\in T$ such that $d_G(s,t) \neq \ell$. In that case $\calP^{(S,T)}$ doesn't contain any path from $s$ to $t$. However, $G^{(S,T)}$ may contain a path from $s$ to $t$. We illustrate an example of this fact in \Fig{cover_example}.

\begin{figure}
  \hspace*{5cm} \includegraphics[width=.35\linewidth]{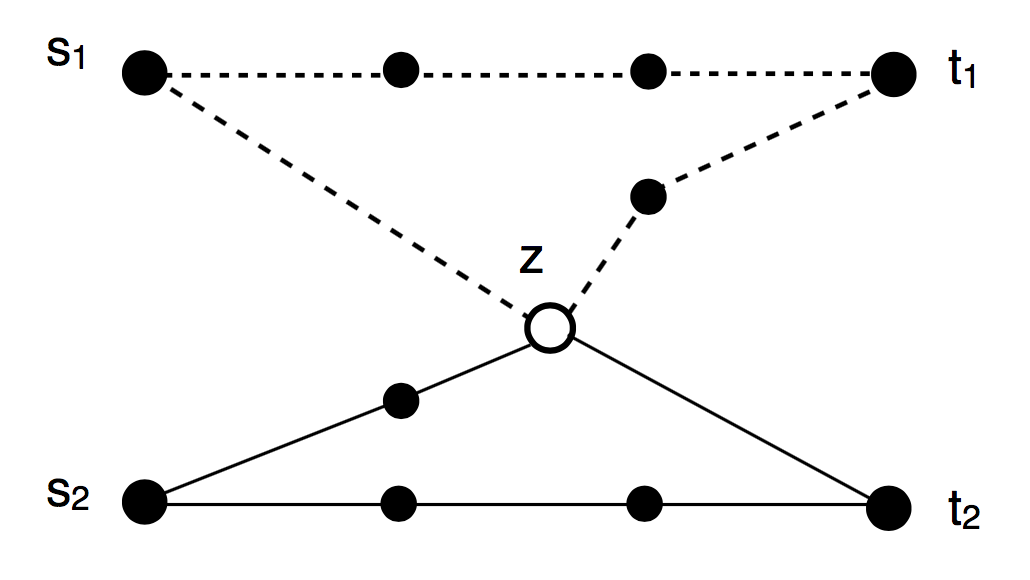}
  \caption{\small{Dashed and solid lines represent the shortest paths from $s_1$ to $t_1$ and $s_2$ to $t_2$, respectively. Let $S = \{s_1,s_2\}$ and $T = \{t_1,t_2\}$. $(S,T)$ is a $3$-consistent pair since $d(s_1,t_1) = d(s_2,t_2) = 3$. Observe that $d(s_1,t_2) = 2$ and $d(s_2,t_1) = 4$ and so the paths from $s_1$ to $t_2$ and from $s_2$ to $t_1$ are not in $\calP^{(S,T)}$; however these paths \textit{will be present} in the cover graph $G^{(S,T)}$. Furthermore, the fact that $z$ lies at step $1$ along a length $3$ shortest path from $s_1$ to $t_1$ \textit{and} at step $2$ along a length $3$ shortest path from $s_2$ to $t_2$ shows that $G^{(S,T)}$ is \textit{not} a $3$-layered DAG and so $(S,T)$ is not $3$-good. However, it is instructive to note that $M^*$ would match $s_1$ to $t_2$ and $s_2$ to $t_1$ since this maximizes $\Psi(\cdot)$.} }
  \label{fig:cover_example}
\end{figure}

An $\ell$-layered DAG is a directed acyclic graph with nodes partitioned into $\ell+1$ layers $L_0,\ldots,L_\ell$ where each edge goes from a vertex in some $L_i$ to a vertex in $L_{i+1}$. A layered DAG is a very structured subgraph.
\begin{definition}[Good Pairs of Consistent Sets]\label{def:good}
	An $\ell$-consistent pair $(S,T)$ is {\em $\ell$-good} if the graph $G^{(S,T)}$ is an $\ell$-layered DAG.
\end{definition}
If $G$ is the hypercube $\{0,1\}^d$, and $(S,T)$ are consistent pairs such that Hamming weight of each vertex in $S$ is the same, and that of each vertex in $T$ is the same (that is $S$ lies in one "level" of the hypercube and so does $T$), then it is easy to see $(S,T)$ is a good pair. 
Lehman and Ron~\cite{LR01} prove that for such $S$ and $T$, one can find $|S| = |T|$ {\em vertex-disjoint} paths. 
Our first lemma is a generalization of the Lehman-Ron theorem~\cite{LR01} to arbitrary good pairs in the augmented-hypergrid.
This, in some sense, abstracts out the sufficient conditions needed for a Lehman-Ron like theorem. We prove this lemma in \Sec{lehmanron}.
\begin{lemma}[Generalized Lehman-Ron]\label{lem:lehmanron}
	If $S,T$ are two subsets of $\AH$ such that $(S,T)$ is an $\ell$-good, consistent pair for some $\ell>0$, then there exists $|S|$ vertex disjoint paths from $S$ to $T$.
\end{lemma}

\begin{definition}[Independent Good Pairs]\label{def:independent}
	Two $\ell$-good pairs $(S,T)$ and $(S',T')$ are said to be independent if any path in $G^{(S,T)}$ is vertex-disjoint from any path in $G^{(S',T')}$.
\end{definition}
\noindent
Independent good pairs are in some sense "far away" from each other -- if we find vertex disjoint paths from $S$ to $T$, and from $S'$ to $T'$, then these paths will not intersect each other.

Recall the definition of the matching $M^*_i$ from the beginning of this section.
Our second lemma shows that if $M^*$ is the $\Psi$-maximizing matching in {\em any} (not necessarily augmented hypergrid) poset, then for all $i$ there is a collection of pairwise independent, $i$-good, consistent pairs of total size $|M^*_i|$. This is proved in \Sec{good}.

\begin{lemma}[Existence of large pairwise-independent, good consistent pairs]\label{lem:good}
Given any poset $G$ and a function $f:G\to \{0,1\}$, let $M^*$ be the maximal cardinality matching with minimum average distance $r$ and maximum $\Psi(\cdot)$ among these. Let $M^*_i$ be the subset of $M^*$ whose endpoints are at distance exactly $i$. Then there exists a collection of pairwise independent $i$-good pairs $(S_1,T_1),\ldots,(S_k,T_k)$ such that $S^*_i = S_1\cup \cdots \cup S_k$ and $T^*_i = T_1 \cup \cdots \cup T_k$.
\end{lemma}

\begin{proof}[Proof of \Thm{vertex}]
	We know there exists $1\leq i\leq 2r$ with $|M^*_i| \geq |M^*|/4r \geq \eps n^d/8r$.
	\Lem{good} gives us a collection $(S_1,T_1),\ldots,(S_k,T_k)$ of pairwise-independent, $i$-good sets for each $i$. \Lem{lehmanron} implies for each $1\leq j\leq k$, we have a collection of $|S_j|$ vertex disjoint paths from $S_j$ to $T_j$. Since they are pairwise independent, the union of these collections is vertex disjoint. This implies we have $\geq \eps n^d/8r$ vertex disjoint paths from a point that evaluates to $1$ to a point that evaluates to $0$. Since each path must contain at least one edge from $S^-_f$, the theorem follows.
\end{proof}

\section{Finding Good Portions in General Posets: Proof of \Lem{good}}\label{sec:good}

We are given the matching $M^*$ which is the maximal cardinality minimum average-distance matching maximizing $\Psi(M^*)$. $M^*_i$ is the subset which looks at pairs exactly at distance $i$. We find the sets $(S_1,T_1),\!\ldots,\!(S_k,T_k)$ via a recursive algorithm (Procedure to Get Pairwise Conflict-free $C_i$'s). To describe the algorithm, we first make a definition.

%
%
%
%
%
%
%
%
\begin{definition} [Conflicting Sets] \label{def:conflict}
Given a pair of disjoint subsets $C,C'\subseteq M^*_i$, let $S:=\{s:\exists (s,t)\in C\}$ and $T:=\{t:\exists(s,t)\in C\}$, and similarly define $S'$ and $T'$. We say that $C$ and $C'$ conflict if there exists {\em shortest} paths $p$ going from some $s\in S$ to some $t\in T$, and $p'$ going from some $s'\in S'$ to some $t'\in T'$ such that (a) $p$ and $p'$ have a vertex $z$ in common, and (b) $d_p(s,z) = d_{p'}(s',z) = j$ and $d_p(z,t) = d_{p'}(z,t') = i - j$.
\end{definition}

Therefore, two sets conflict if there are shortest paths from their respective $(S,T)$'s intersecting "at \emph{the same} level". Note that the paths needn't be from $s$ to $M^*(s)$, nor do we say the sets conflict if the paths intersect but "at \emph{different} levels". However, as seen later in this section, the pairwise conflict-free sets we obtain from $M^*$ via our recursive algorithm indeed have pairwise disjoint cover graphs as otherwise we would obtain another matching with either (a) smaller average distance or (b) the same average distance and larger $\Psi(\cdot)$, contradicting our definition of $M^*$. The following procedure returns a collection of subsets $C_1,\ldots,C_k$ of $M^*_i$, such that they are pairwise conflict-free. The sets
$S_i,T_i$ are obtained by taking the lower and upper endpoints of $C_i$.

%
%
%
%
%

\begin{framed} \label{procedure}
	\noindent
	\textbf{Procedure to Get Pairwise Conflict-free $C_i$'s:} \\

\noindent Suppose $M^*_\ell = \{(x_1,y_1),(x_2,y_2),...,(x_m,y_m)\}$. The procedure is recursively defined:\smallskip

\noindent \textbf{Base Step:} Define the leaves of $\mathcal{H}$ as $C^{(0)}_{1} = \{(x_1,y_1)\}, C^{(0)}_{2} = \{(x_2,y_2)\}, ..., C^{(0)}_{m} = \{(x_m,y_m)\}$. Construct the base \textit{conflict graph} $G^{(0)}$ as follows: each $C^{(0)}_{i}$ is a vertex and $C^{(0)}_{i}$ is connected by an edge to $C^{(0)}_{j}$ if they conflict (\Def{conflict}) and $i \neq j$. Exit if $G^{(0)}$ has no edges.\smallskip
%

\noindent \textbf{Recursive Step:} For $i\geq 0$:
for the $j$th connected component in $G^{(i)}$, construct a set $C^{(i+1)}_j$ which is the union of all $C^{(i)}_k$'s in the $j$th connected component. Construct the graph $G^{(i+1)}$ on these nodes indexed by $C^{(i+1)}_j$'s. Exit if $G^{(i+1)}$ has no edges.
Note that the number of nodes in $G^{(i+1)}$ is strictly less than that in $G^{(i)}$ since the latter has at least one edge. 
Also note $G^{(i+1)}$ may have new edges since the conflict sets are getting bigger.\smallskip

\noindent \textbf{Termination:} Since the number of vertices in the $G^{(i)}$'s strictly decrease, this procedure terminates.
Let $C^{(\omega)}_1,C^{(\omega)}_2,\ldots,C^{(\omega)}_k$ be the collection of sets at this level. By definition these are pairwise conflict free. Let $S_i$ (resp, $T_i$) be the lower (resp, upper) set of endpoints of the pairs in $C^{(\omega)}_i$. \smallskip

\noindent {\bf Return} $(S_1,T_1),\ldots,(S_k,T_k)$.

%
%
%
%
\end{framed}
First note that the lower and upper endpoints of any set $C^{(\alpha)}_j$ is $i$-consistent since they can be paired using the matching $M^*_i$. Also note that at any iteration $\alpha$, every matched pair $(s,t)\in M^*_i$ is in some $C^{(\alpha)}_j$. Therefore, the $S_i$'s partition $S^*_i$ and similarly $T_i$'s partition $T^*_i$. What remains to be proven is that (a) each $(S_i,T_i)$ is good, and (b) they are pairwise independent. Before we do so we need the following "rematching lemma" which is key.

Fix one of the sets $C^{(\omega)}_j$ in the conflict-free collection, and let $(S_j,T_j)$ be the sets obtained. For the sake of the rematching lemma let us forsake the subscript $j$. 
%

\begin{lemma} [Rematching Lemma] \label{lem:rematch}
 For any $\hat{s} \in S$, $\hat{t} \in T$, it is possible to rearrange $M^*_i$ to form a new matching $M'$ with the following properties:

\begin{itemize}[noitemsep]
    \item For any $s \in S$ with $s \neq \hat{s}$: $d(s,M'(s)) = i$.
    \item $\hat{s}$ and $\hat{t}$ are the only vertices which become unmatched. 
\end{itemize}
\end{lemma}

We defer the proof of the rematching lemma and first note how it helps us.
Once again, since $S$ and $T$ arise from $C^{(\omega)}_j$, they are $i$-consistent and indeed $(s,M^*(s))$ is the pairing. 
What the above lemma says is that for {\em any} $\hat{s}$ and $\hat{t}$, we can "rewire" the matching to $M'$ so that the average distance still remains $i$. In particular, if $d(\hat{s},\hat{t}) < i$, we would have a contradiction since we would get a different maximal matching with strictly less distance. With this in mind, let's use the rematching lemma to prove that the $(S_i,T_i)$'s are good and pairwise independent (\Def{independent}) thus proving \Lem{good}.

\begin{lemma}\label{lem:wrapup}
    Let $(S_1,T_1), \ldots, (S_k,T_k)$ be the pairs of sets returned by "Procedure to Get Pairwise Conflict-free $C_i$'s" on input $M^*_i$.
	Each $(S_j,T_j)$ is $i$-good. Moreover, $(S_j,T_j)$ and $(S_{j'},T_{j'})$ are independent for $j \neq j'$.
\end{lemma}
\begin{proof}

		Recall the definition of the cover graph -- we need to show $G^{(S_j,T_j)}$ is a layered DAG. 
		Suppose not. Then there must exist a vertex $z$ which (a) lies on a path $p$ from $\hat{s}\in S$ to $t\in T$, (b) also lies on a path $p'$ from $s\in S$ to $\hat{t}\in T$ where both paths are of length $i$, are shortest paths between their endpoints,
		and (c) $d_p(\hat{s},z) \neq d_{p'}(s,z)$. The situation is illustrated as follows:
$$\hat{s} \xrightarrow[]{a} z \xrightarrow[]{i-a} t$$ and $$s \xrightarrow[]{b} z \xrightarrow[]{i-b} \hat{t} \text{.}$$ 
		where we assume wlog that $a<b$.
		
		Now, by \Lem{rematch} (rematching lemma), there exists a rearrangement $M'$ of the endpoints of $(S_j,T_j)$ 
		such that $\forall s\neq \hat{s}, d(s,M'(s))=i$, and only $\hat{s}$ and $\hat{t}$ are unmatched. 
		However, $\hat{s}$ and $\hat{t}$ have a path of length $=a+i-b < i$. Therefore, we can add $(\hat{s},\hat{t})$ to $M'$ to obtain a matching whose average length is strictly smaller than that of $M^*$. Contradiction. Therefore $(S_j,T_j)$ must be $i$-good.\smallskip
		
		We now claim $(S_1,T_1)$ and $(S_2,T_2)$ are independent. Suppose not, and there is a shortest path $p_1$ from $s_1\in S_1$ to $t_1\in T_1$ which intersects a shortest path $p_2$ from $s_2\in S_2$ to $t_2\in T_2$. Suppose $z$ is the first (nearest to the $s$'s) at which they meet. Since each $(S_j,T_j)$ is good, the graph $G^{(S_j,T_j)}$ is layered, and therefore these paths have to be shortest paths of length $i$.
		Two cases arise: (a) $d_{p_1}(s_1,z) = d_{p_2}(s_2,z)$. But then that would mean the corresponding $C^{(\omega)}_1$ and $C^{(\omega)}_2$ conflict. Contradiction. (b) $d_{p_1}(s_1,z) < d_{p_2}(s_2,z)$. Again apply the rematching \Lem{rematch} to get two rewired matchings $M'_1$ and $M'_2$ which leave $s_1,t_1$ and $s_2,t_2$ unmatched while all the other pairs are at distance $i$. Now add the pairs $(s_1,t_2)$ and $(s_2,t_1)$ in the matching. Observe this has a larger $\Psi()$ since we replaced two pairs at distance $i$ with two pairs with unequal distances summing to $2i$. Contradiction again.
In sum, all the $(S_j,T_j)$'s are pairwise independent and good.
%
%
\end{proof}
\subsection{Proof of the Rematching \Lem{rematch}}
Just for this proof, we use $M^*$ without the superscript. This is purely for brevity's sake. \Fig{rematch} accompanies the inductive step. 

Let $a$ be the smallest index such that $(\hat{s},M(\hat{s}))$ and $({M}^{-1}(\hat{t}), \hat{t})$ lie in the same $C^{(a)}_\ell$ for some $\ell$. We know $1\leq a\leq \omega$. We prove by induction on $a$.\smallskip

\noindent \textbf{Base Case:} $a=1$. Since $(\hat{s},M(\hat{s})),(M^{-1}(\hat{t}), \hat{t}) \in C^{(1)}_\ell$, we know there is a path from $\{(\hat{s},M(\hat{s}))\}$ to $\{(M^{-1}(\hat{t}), \hat{t})\}$ in $G^{(0)}$. Suppose that the length of the shortest such path is $q$. Let this
path be from $C^{(0)}_0 = \{(\hat{s},M(\hat{s}))\}$ to $C^{(0)}_q = \{(M^{-1}(\hat{t}),\hat{t})\}$.
Let the $j$th node in this path be $C^{(0)}_j = \{(s_j,M(s_j))\}$. 
By definition, $\{(s_j,M(s_j))\}$ conflicts with $\{(s_{j+1},M(s_{j+1}))\}$. Therefore, we can rewire the matching $M'$ which maps for all $1\leq j\leq q$, $s_j$ to $t_{j-1}$. By the definition of conflict, each of these pairs are at distance exactly $i$. And $\hat{s}$ and $\hat{t}$ are the ones left unmatched. On the rest of the pairs, $M'$ and $M$ agree. \smallskip
%

\noindent \textbf{Inductive Step:} Since $a$ is the smallest value such that $(\hat{s},M(\hat{s})),(M^{-1}(\hat{t}), \hat{t}) \in C^{(a)}_\ell$ we know there are sets $C^{(a-1)}_j$ and $C^{(a-1)}_{j'}$ such that they're disjoint, and $(\hat{s},M(\hat{s})) \in C^{(a-1)}_j$ and $(M^{-1}(\hat{t}),\hat{t}) \in C^{(a-1)}_{j'}$. 
Moreover, by construction of $C^{(a)}_\ell$ we know that there is a path from $C^{(a-1)}_{j}$ to $C^{(a-1)}_{j'}$ in the \textit{conflict} graph, $G^{(a-1)}$. Let the shortest such path be of length $q$ and let the shortest path be $C^{(a-1)}_{j} = C_0, C_1, \cdots, C_q = C^{(a-1)}_{j'}$. Let $S_k = \{x | (x,y) \in C_k\}$ and $T_k = \{y | (x,y) \in C_k\}$ for $0 \leq k \leq q$.

For any $k$ in the range $0 \leq k < q$, the sets $C_k$ and $C_{k+1}$ conflict.
Therefore, we know there exists $s_k \in S_k, t_k \in T_k$ and $s_{k+1} \in S_{k+1}, t_{k+1} \in T_{k+1}$ and $z_{k+1}$ such that $$s_k \prec z_{k+1} \prec t_k \text{ and } s_{k+1} \prec z_{k+1} \prec t_{k+1} \text{.}$$
Also, we have $d(s_k,z_{k+1}) = d(s_{k+1},z_{k+1})$ and $d(z_{k+1},t_k) = d(z_{k+1},t_{k+1})$. Thus $d(s,t) = i$ for all combinations of $s \in \{s_k,s_{k+1}\}$ and $t \in \{t_k,t_{k+1}\}$.

By induction on $C_0$ (which recall is $C^{(a-1)}_j$), we can rearrange $M \cap C_0$ to get $M'$ where $\hat{s}$ and $t_0$ are the only unmatched endpoints from $C_0$ (since $\hat{s}$ and $t_0$ are both endpoints from $C_0$). Now, for all $k$ in the range $1 \leq k< q$, by induction on $C_k$, we can rearrange $M \cap C_k$ to get $M'$ where $s_k$ and $t_k$ are the only unmatched endpoints from $C_k$. Finally, by induction on $C_q$, we can rearrange $M \cap C_q$ to get $M'$ where $s_q$ and $\hat{t}$ are the only unmatched endpoints from $C_q$. Our matching is now $M'$ where $\forall (x,y) \in M'$, $d(x,y) = i$ and the sets of unmatched endpoints are $\{\hat{s} = s_0, s_1, s_2, ..., s_q\}$ and $\{t_0,t_1,...,t_{q-1},t_q = \hat{t}\}$. By the existence of $z_k$ for $1 \leq k \leq q$ we can set $M'(s_k) = t_{k-1}$ for all $k$ in the range $1 \leq k \leq q$. Moreover, $d(s_k, t_{k-1}) = i$ for all $k$. The only remaining unmatched endpoints are $\hat{s}$ and $\hat{t}$. This completes the proof of the rematching \Lem{rematch}.

\begin{figure}
  \hspace*{1.5cm} \includegraphics[width=.75\linewidth]{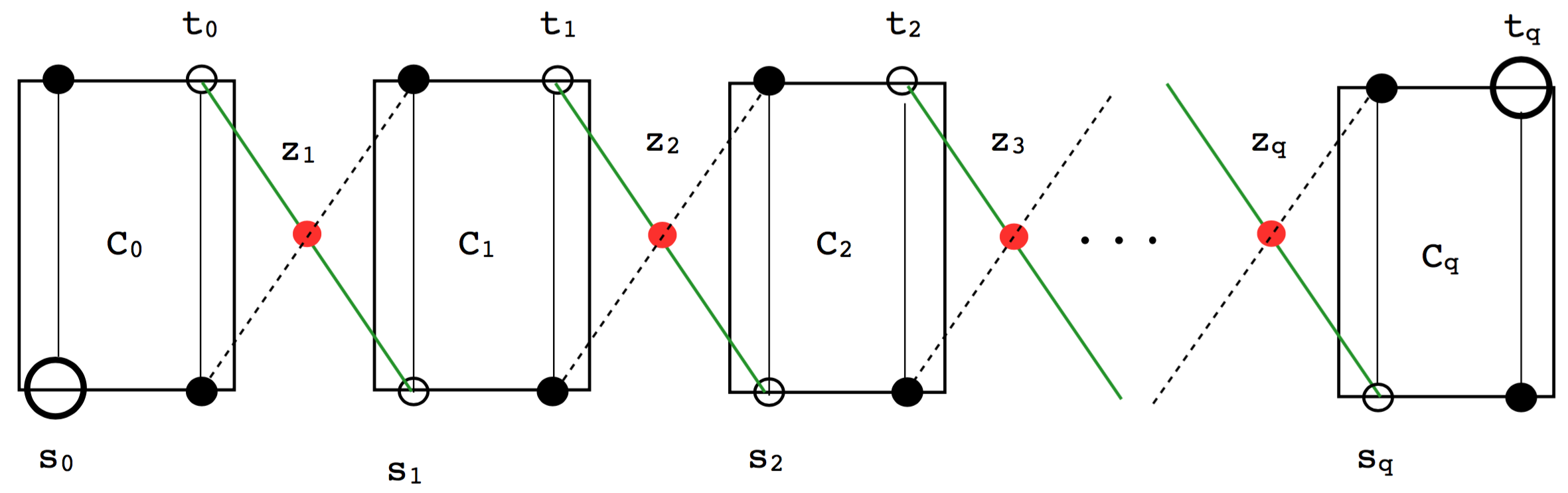}
  \caption{\small{An accordion-like structure illustrating the proof of the rematching lemma. Unfilled circles represent vertices left unmatched by inductively invoking the lemma on each $C_k$. Solid black, vertical lines represent the matching $M$ while solid green, diagonal lines represent the newly matched pairs of $M'$. Dashed, diagonal lines illustrate the additional length $i$ shortest paths passing through each $z_k$. This accordion-like structure recursively takes place within each $C_k$. Note that $s_0 = \hat{s}$ and $t_q = \hat{t}$.}}
  \label{fig:rematch}
\end{figure}

\section{Routing on the Augmented Hypergrid: Proof of~\Lem{lehmanron}}\label{sec:lehmanron}

In this section we prove the generalization of the routing theorem of Lehman-Ron~\cite{LR01} for good pairs $(S,T)$ in $\AH$. 
This proof is akin to the proof in~\cite{LR01}.

Suppose $(S,T)$ is a $\ell$-good consistent pair in $\AH$ with $|S| = |T| = m$. We show that there exists $m$ vertex disjoint paths from $S$ to $T$ in the $\ell$-cover $\AH^{(S,T)}$. Since $(S,T)$ are $\ell$-consistent, there is a bijection $\phi: S \to T$ 
with $d(s,\phi(s))=\ell$ (recall definitions \ref{def:consistent}, \ref{def:cover} and \ref{def:good}). The proof is by induction on $\ell$ and $m$. The base cases are trivial. If $m = 1$, then any path we choose from $S$ to $T$ suffices. If $\ell = 1$, then $\phi$ immediately gives a matching of edges in $\AH$ and this gives us our vertex disjoint paths.

Since $(S,T)$ is good, $\AH^{(S,T)}$ is a layered graph. Let $L_i$ be its $i$'th layer for $0\leq i\leq \ell$.
For a vertex $u \in \AH^{(S,T)}$, let $\delta^-(u),\delta^+(u)$ denote the in and out degree of $u$ in $\AH^{(S,T)}$. 

\begin{claim}\label{clm:degree}
    If $v$ is reachable from $u$ in $\AH^{(S,T)}$, then $\delta^+(u) \geq \delta^+(v)$ and $\delta^-(u) \leq \delta^-(v)$.
\end{claim}

\begin{proof} 
	
	For $v \in [n]^d$, let $v \oplus_i a = (v_1, ..., v_{i-1}, v_i + a, v_{i+1}, ..., v_d)$ and $v \oplus_i -a = (v_1, ..., v_{i-1}, v_i - a, v_{i+1}, ..., v_d)$. For notational convenience in the following proof we let $V$ and $E$ stand for the vertex and edge sets of $\AH^{(S,T)}$.

To establish $\delta^+(u) \geq \delta^+(v)$ we show $(v, v \oplus_i a) \in E$ \textit{implies} $(u, u \oplus_i a) \in E$. Suppose $u \in L_i$ and $v \in L_j$ where $i < j$. Since, $u,v \in V$ we know there is some $s \in S$ and $t \in T$ such that $d(s,u) = i$, $d(u,v) = j-i$ and $d(v,t)=\ell-j$. Since $v \oplus_i a \in V$ there is a path from $v \oplus_i a$ to some $t' \in T$ of length $\ell - j - 1$. Clearly, $d(u \oplus_i a,v \oplus_i a) = j - i$ and $d(s,u \oplus_i a) = i + 1$. This gives a path from $s$ to $t'$ of length $\ell$. Finally, we cannot have $d(s,t') < \ell$ since this would contradict the fact that $\AH^{(S,T)}$ is a $\ell$-layered DAG. That is, this would imply there is an edge in $E$ joining a vertex in some $L_k$ to $L_{k'}$ where $k' > k - 1$. Thus, $(u, u \oplus_i a)$ lies on a shortest path of length $\ell$ from $S$ to $T$ and so $(u, u \oplus_i a) \in E$.

Similarly, $(u \oplus_i -a, u) \in E$ \textit{implies} $(v \oplus_i -a, v) \in E$ and so $\delta^-(u) \leq \delta^-(v)$. The proof is analogous to the previous paragraph and so is ommitted. \end{proof}
\noindent
We make use of \Clm{degree} to show that there must exist a layer of $\AH^{(S,T)}$ with size at least $m$. 

\begin{claim}\label{clm:routing_claim}
    $|L_1| \geq m$ \textit{or} $|L_{\ell-1}| \geq m$.
\end{claim}

\begin{proof} Suppose $|L_{\ell-1}| = m' < m$. Pick \textit{any} $m'$ vertices in $S$ to get $S'$ and let $T' := \phi(S')$ be the image of $\phi$ on $S'$. Clearly, $(S',T')$ is $\ell$-consistent by virtue of the bijection $\phi$. Moreover, $\AH^{(S',T')}$ is a subgraph of $\AH^{(S,T)}$, which is an $\ell$-layered DAG since $(S,T)$ is $\ell$-good. Thus, $\AH^{(S',T')}$ must also be $\ell$-layered and so $(S',T')$ is $\ell$-good. Thus, by induction there exists $m'$ vertex disjoint paths from $S'$ to $T'$. This induces a bijection $\psi: S' \to L_{\ell-1}$ such that the path beginning at $s \in S'$ contains the vertex $\psi(s)$. We now have the following inequality: 

\begin{align}\label{routing_proof:ineq_1}
    \delta^+(S) = \sum_{s \in S} \delta^+(s) > \sum_{s \in S'} \delta^+(s) \geq \sum_{s \in S'} \delta^+(\psi(s)) = \sum_{v \in L_{\ell-1}} \delta^+(v) = \delta^-(T) \text{.}
\end{align}

The first inequality holds because every vertex in $S$ has positive out-degree (by \Def{consistent}). The second inequality holds by \Clm{degree}. The second to last equality holds since $\psi$ is a bijection. The final equality is because an edge $(u,v)$ satisfies $u \in L_{\ell-1}$ if and only if $v \in L_\ell = T$ since $\AH^{(S,T)}$ is $\ell$-layered. 


Now, suppose $|L_{1}| = m' < m$. In a similar fashion, pick \textit{any} $m'$ vertices in $S$ to get $S'$ and again let $T' := \phi(S')$. By induction there exists $m'$ vertex disjoint paths from $S'$ to $T'$ and this induces a bijection $\psi': L_{1} \to T'$. Through an analogous argument we get

\begin{align} \label{routing_proof:ineq_2}
    \delta^-(T) = \sum_{t \in T} \delta^-(t) > \sum_{t \in T'} \delta^-(t) \geq \sum_{t \in T'} \delta^-(\psi^{-1}(t)) = \sum_{u \in L_{1}} \delta^-(v) = \delta^+(S) \text{.}
\end{align}

Hence, $|L_{\ell-1}| < m$ implies $\delta^+(S) > \delta^-(T)$, while $|L_1| < m$ implies $\delta^+(S) < \delta^-(T)$ and so we have a contradiction when both are true. Therefore, either $|L_1| \geq m$ or $|L_{\ell-1}| \geq m$. \end{proof}

Now, let $\widehat{L}$ denote $L_i$ for which $|L_i| \geq m$ and $i \in \{1,m-1\}$, which exists due to \Clm{routing_claim}. Define the tripartite graph $G_{S,\widehat{L},T}$ where for $s \in S$, $v \in \widehat{L}$, $t \in T$, $(s,v)$ is an edge if $v$ is reachable from $s$ in $\AH^{(S,T)}$ and $(v,t)$ is an edge if $t$ is reachable from $v$ in $\AH^{(S,T)}$. 

Suppose for a moment that there are $m$ vertex disjoint paths in $G_{S,\widehat{L},T}$ from $S$ to $T$. This induces a 3-dimensional matching $R = \{(s,v,t) : s \in S, v \in V, t \in T\}$ of size $m$ such that $(s,v,t) \in R$ means that $s$ is routed to $t$ by a path which contains $v$. Furthermore, $V \subseteq \widehat{L}$ with $|S| = |V| = |T| = m$. Now define $\phi_1:S \to V$ and $\phi_2:V \to T$ as $\phi_1(s) = v$, $\phi_2(v) = t$ when $(s,v,t) \in R$. Observe that the existence of $\phi_1$ and $\phi_2$ shows that $(S,V)$ and $(V,T)$ are $i$-consistent and $(\ell-i)$-consistent pairs, respectively. Moreoever, notice that any shortest path of length $i$ from $S$ to $V$ in $\AH^{(S,T)}$ is a subpath of some shortest path of length $\ell$ from $S$ to $T$ in $\AH^{(S,T)}$. Thus, $\AH^{(S,V)}$ is a subgraph of $\AH^{(S,T)}$. Additionally, it's easy to see that the $j$'th layer of $\AH^{(S,V)}$ is a subset of the $j$'th layer of $\AH^{(S,T)}$ for any $j \in [i]$. Thus, $(S,V)$ is $i$-good in $\AH$ and an analogous argument shows $(V,T)$ is $(\ell-i)$-good in $\AH$.

By induction, there are $m$ vertex disjoint paths from $S$ to $V$ and $m$ vertex disjoint paths from $V$ to $T$. Stitching these paths together yields $m$ vertex disjoint paths from $S$ to $T$. \\

Now, we prove there must be $m$ vertex disjoint paths in $G_{S,\widehat{L},T}$. Suppose, for the sake of contradiction, this is not true. Then, by Menger's theorem there exists a cut $C$ separating $S$ from $T$ in $G_{S,\widehat{L},T}$ with $|C| < m$. \\

\noindent - \textit{Case (1):} $C \cap S = C \cap T = \emptyset$. So $C \subseteq \widehat{L}$. Recall that for any $v \in \widehat{L}$ there exists $s \in S, t \in T$ such that $(s,v)$ and $(v,t)$ are edges in $G_{S,\widehat{L},T}$ (since $\widehat{L}$ is simply a set of vertices in the cover graph $G^{(S,T)}$ and $\widehat{L} \cap S = \widehat{L} \cap T = \emptyset$). Thus if $C \neq \widehat{L}$, then $C$ would not be a cut. Thus, $C = \widehat{L}$, but $|\widehat{L}| \geq m$ and we have a contradiction. \\

\noindent - \textit{Case (2):} $C \cap S \neq \emptyset$ and/or $C \cap T \neq \emptyset$. Let $S' \subset S$ be the set of vertices not in $C$ which are \textit{not} mapped by $\phi$ to vertices in $C$ and let $T' = \phi(S')$. Note that $S' \cap C = T' \cap C = \emptyset$. Let $\widehat{L}' \subseteq \widehat{L}$ denote the set of vertices on any path from $S'$ to $T'$. We have $|S'| \geq m - (|C \cap S| + |C \cap T|) > |C| - (|C \cap S| + |C \cap T|) = |C \cap \widehat{L}|$. By induction (on $m$) there are $|S'|$ vertex disjoint paths from $S'$ to $T'$ in $G'_{S',\widehat{L}',T'}$ (defined analogously to $G_{S,\widehat{L},T}$) and so $|\widehat{L}'| \geq |S'|$. Thus, $|\widehat{L}'| > |C \cap \widehat{L}|$ and so there exists a vertex in $\widehat{L}'$ that is not in $C$, contradicting the assumption that $C$ separates $S$ from $T$.\\

\noindent
This completes the proof of~\Lem{lehmanron}.

\section{Analysis of the Tester}\label{sec:tester_analysis}

Given our Margulis-type theorem (\Thm{main-margulis}), one can obtain a $o(d) \cdot \polylog~n$-tester by modifying either Chakrabarty-Seshadhri~\cite{ChSe13-j}, Chen et al~\cite{ChenST14}, or Khot-Minzer-Safra~\cite{KMS15}. 
The last, in our opinion, provides the cleanest analysis.
We borrow their techniques to analyze our tester. 

Recall that for each dimension $i \in [d]$, our tester chooses some edge matching $H_i := H^{c_i}_{i,a_i}$ from the set $\mathbf{H}$. 
We remind the reader that $S_f$ is the set of edges $(x,y)$ in $\AH$ where $f(x) \neq f(y)$. Moreover, $S_f^+ = \{(x,y) \in S_f : f(x) = 0, f(y) = 1\}$ and $S_f^- = \{(x,y) \in S_f : f(x) = 1, f(y) = 0\}$. The total, positive and negative influences are denoted respectively by $I_f := |S_f|/n^d$, $I^+_f := |S^+_f|/n^d$ and $I^-_f := |S^-_f|/n^d$.

We will need the following lemma that bounds the total influence in $\AH$, when
the negative influence is not too large.
The corresponding theorem for Boolean hypercubes (Theorem 9.1 in \cite{KMS15}) is easy; we need to work a bit harder for $\AH$. 
We defer the proof to \Sec{influence}.

\begin{lemma} \label{lem:influence} 
    If $I_f^- < \sqrt{d}$, then $I_f < 7\sqrt{d}\log n$. 
\end{lemma}

\noindent We first make some observations before taking on the main analysis.

\begin{observation} [Edge Tester] \label{obs:edge_tester} 
    With probability $\Theta(\frac{1}{\log d})$ our tester chooses $\tau = 1$. Thus our tester is the edge tester on $\AH$ described by \cite{ChSe13} with probability $\Theta(\frac{1}{\log d})$. 
\end{observation}

\begin{observation} [Total Influence Bound] \label{obs:inf} 
    If $|I_f^-| \geq \sqrt{d}$ then the edge tester detects a violation with probability $\Omega\Big(\frac{1}{\sqrt{d}\log n}\Big)$  since $\AH$ contains $\Theta(n^d \cdot d \log n)$ edges. This, combined with the previous observation proves \Thm{main-theorem} for the case $|I_f^-| \geq \sqrt{d}$. Thus, we assume $I_f^- < \sqrt{d}$. By \Lem{influence},
$I_f < 7\sqrt{d}\log n$. 
\end{observation}

\smallskip

We use the following definition of $\tau$-persistence that is nearly identical to the definition given by \cite{KMS15}.

\begin{definition}\label{def:persistence} Fix $\tau$ and $x \in \AH$.
Consider the tester selecting the parameter $\tau$ in~\Step{tau} of \Fig{alg}.
We define $x$ to be $\tau$-\textit{persistent} if $Pr_y[f(x) \neq f(y)] \leq \frac{1}{10} \text{.}$
(The probability is over the distribution of $y$ as defined in~\Step{y} of \Fig{alg}.)
\end{definition}

\noindent
We will need the following lemma, which is obtained from the proof of Lemma 9.3 in~\cite{KMS15}.

\begin{lemma} [Paraphrased from the proof of Lemma 9.3 in~\cite{KMS15}] \label{lem:kms} Consider a Boolean function $g:\{0,1\}^r \rightarrow \{0,1\}$.
Let $\tau \in [1,\sqrt{r/\log r}]$. For any sufficiently large $B$, there exists 
constant $\beta > 0$ such that the following holds. Pick a uniform random vertex $x$
with Hamming weight in $[r/2 - B \sqrt{r\log r}, r/2 + B \sqrt{r\log r}]$.
Let $y$ be a random vertex obtained by flipping $\tau$ random zeros of $x$ to $1$. Then 
$\Pr_{x,y}[g(x) \neq g(y)] \leq \beta \tau I_g/r$. \smallskip

In particular, suppose $B$ is so chosen so that the probability a uniform at random $x$ has Hamming weight in $[r/2 - B \sqrt{r\log r}, r/2 + B \sqrt{r\log r}]$ is at least $1-1/r^{10}$. Given a uniform at random $x$, if $x$ has less than $\tau$ zeros, define $y=x$. Otherwise, obtain $y$ as before. Then $\Pr_{x,y}[g(x) \neq g(y)] \leq \beta \tau I_g/r + \frac{1}{r^{10}} $
\end{lemma}

We prove our main lemma for hypergrids, by applying \Lem{kms}
to randomly chosen hypercubes in the hypergrid.


\begin{lemma}\label{lem:persistence} Let $\tau \in [1,\sqrt{d/10\log d}]$ be any integer. The fraction of $\tau$-non-persistent vertices in $\AH$ is at most 
$O\big(\tau \cdot \frac{I_f}{d \log n} + d^{-9}\big)$.
\end{lemma}
\begin{remark}\emph{
		Readers familiar with~\cite{KMS15} will note that the extra $d^{-9}$ term is absent in Lemma 9.3 in their paper. 
		This is because the authors implicitly assume that the point $x$ lies in the ``middle layers'' of the hypercube. 
		\Lem{kms} makes it explicit which leads to the correction factor in \Lem{persistence} above. 
		Indeed this correction is needed even in~\cite{KMS15} -- consider the function $f$ over the Boolean hypercube which is $1$ at the all-ones vector and $0$ everywhere else. The total influence $I_f = d/2^d$. The vertices having exactly $\tau$ zeros are non-persistent and so
		the fraction of $\tau$-non-persistent vertices is ${d\choose \tau} / 2^d \approx d^\tau/2^d$, while $\tau I_f/d = \tau/2^d$.
		However, if $f$ is $\eps$-far from monotone, then $I_f \geq \eps$ and thus the first term subsumes the correction factor.
		We make it explicit since $n$ could in general be $\gg d$.
	}
\end{remark}

\begin{proof}[{\bf Proof of \Lem{persistence}}] 
Let $\cU$ denote the uniform distribution over $[n]^d$. We will show
\begin{equation}
\label{eq:bonda}
\EX_{x \sim \cU} \big[\Pr_y [f(x) \neq f(y)]\big] \leq O\left(\tau \cdot \frac{I_f}{ d \log n}\right) + O(d^{-9})
\end{equation}
where the probability over $y$ is the distribution obtained by \Step{y} of \Fig{alg}.
This will prove the lemma since $\Pr_y [f(x) \neq f(y)] > 1/10$ for $\tau$-non-persistent vertices. \smallskip

Observe that $y$ is chosen by first choosing the random collection $\cH$ of matchings in \Step{H}, and then (possibly)
choosing $T$ in \Step{y}. For a fixed $x$, let $\chi(x,\cH,T)$ denote the indicator for the event $f(x) \neq f(y)$ given $\cH$ and $T$. 
Therefore, we get 
\begin{equation}\label{eq:laddoo}
\EX_{x \sim \cU}\Pr_y[f(x)\neq f(y)] = \EX_{x \sim \cU}\EX_{\cH} \EX_{T} [\chi(x,\cH,T)]
\end{equation}
\noindent
Given $\cH$, consider the DAG on $[n]^d$ obtained by adding all directed edges $(u,v) \in H_i$ for all $H_i \in \cH$. Observe that this partitions the $[n]^d$ into connected components, each of which is a hypercube. In particular, if $x\in [n]^d$ participates in $k$ of the matchings in $\cH$, then the connected component containing $x$ is a $k$-dimensional hypercube. To see this, observe that if $x$ participates in $H_1$ and $H_2$, then if $y_1$ and $y_2$ are the points $x$ is matched to in $H_1$ and $H_2$ respectively, then $y_1$ also participates in $H_2$ and $y_2$ also participates in $H_1$, and indeed are matched to the same point.

Let $\cD_{\cH}$ be the distribution on these hypercubes, where each cube is chosen with probability proportional to the number of vertices it contains.
For any such cube $C$, let $\cU_C$ denote the uniform distribution over all vertices in $C$.
Note that $x\sim \cU$ can be obtained by first sampling $C\sim \cD_{\cH}$ and then sampling $x\sim \cU_C$. Therefore, we get that the RHS of \eqref{eq:laddoo} is
\begin{eqnarray}
    \EX_{x \sim \cU} \EX_{\cH} \EX_{T} [\chi(x,\cH, T)] = \EX_{\cH} \EX_x \EX_{T} [\chi(x, \cH, T)]
    = \EX_{\cH} \EX_{C \sim \cD_{\cH}} \EX_{x \sim \cU_C} \EX_T [\chi(x,\cH,T)]  \label{eq:x-switch}
\end{eqnarray}
Let us now analyze $\EX_{x \sim \cU_C} \EX_T [\chi(x,\cH,T)]$. 
Let $f_{|C}$ be the function restricted to the points in the hypercube $C$. Note that $\chi(x,\cH,T)=1$ if and only if $f_{|C}(x) \neq f_{|C}(y)$ where $y = x$ if $x$ has less than $\tau$ zeros, and otherwise is obtained 
by flipping $\tau$ random zeros of $x$ in $C$ to $1$.
This is exactly the random process described in \Lem{kms} above. Therefore, we get 
\[
\EX_{x \sim \cU_C} \EX_T [\chi(x,\cH,T)] \leq O\left(\tau \cdot \frac{I_{f_{|C}}}{\dim(C)}\right) + \dim(C)^{-1/10}
\]
where $\dim(C)$ is the dimension of the cube $C$. We now break into two cases: if $\dim(C) < d/4$, we use the trivial upper bound of $1$, otherwise we use the inequality above. Using $\psi(C)$ to be the indicator that $\dim(C) < d/4$, plugging into \eqref{eq:x-switch} we get
\begin{eqnarray}
    \EX_{x \sim \cU} \EX_{\cH} \EX_{T} [\chi(x,\cH, T)] & \leq &\EX_{\cH} \EX_{C \sim \cD_{\cH}}[\psi(C) + O(\tau I_{f|_C}/d) + O(d^{-10})] \nonumber \\
    & = & \EX_{\cH} \EX_{C \sim \cD_{\cH}} [\psi(C)] + O(\tau/d)\EX_{\cH} \EX_{C \sim \cD_{\cH}} [I_{f|_C}] + O(d^{-10}) \label{eq:break}
\end{eqnarray}
\noindent
Each matching $H \in \cH$ is of the form $H^{c_i}_{i,a_i}$ and is a perfect matching if $c_i = 0$.
If more than $d/4$ of the $d$ possible $c_i$'s are set to $0$, then every cube $C$ in the partition of $[n]^d$ so obtained will have dimension at least $d/4$. Thus $\psi(C) = 0$ for all $C$. The probability that less than $d/4$ of the $c_i$'s are set to $0$, by a Chernoff bound, is $\leq 2^{-d/10}$. Therefore, $\EX_{\cH} \EX_{C \sim \cD_{\cH}} [\psi(C)] \leq 2^{-d/10}$.
%
%

We now deal with the second term in \eqref{eq:break}. For convenience, for any dimension $i$ matching $H_i$,
let $x + H_i$ be the upper endpoint of the $H_i$-edge containing $x$ as a lower
endpoint, if this edge exists. If not, let $x+H_i$ be $x$.
\begin{eqnarray}
\EX_{\cH} \EX_{C \sim \cD_{\cH}} [I_{f|_C}]  =  \EX_{\cH} \EX_{C \sim \cD_{\cH}} \sum_{i \in [d]} \EX_{x \sim \cU_C} [\bone_{f(x) \neq f(x+H_i)}] 
 =  \sum_{i \in [d]} \EX_{H_i} \EX_{x \in \cU} [\bone_{f(x) \neq f(x+H_i)}] \label{eq:inf}
\end{eqnarray}
Observe that $I_f = \sum_{i \in [d]} \sum_{H_i} \EX_{x \in \cU} [\bone_{f(x) \neq f(x+H_i)}]$.
Since there are $2\log n$ choices of $H_i$, the expression in \Eqn{inf} is precisely $I_f/(2\log n)$.
Putting it all together, we get $\EX_{x \sim \cU} \EX_{\cH} \EX_{T} [\chi(x,\cH, T)]\leq O(\tau I_f/(d\log n)) + O(d^{-10}) + 2^{-d/10}$.
Noting that the last two terms add up to $\leq d^{-9}$ for large enough $d$, we are done due to~\eqref{eq:bonda} and \eqref{eq:laddoo}. 
\end{proof}

\subsection{Main Analysis of the Tester}

\noindent
We are now equipped to prove the following \Lem{detect}, which is the main analysis of our tester. \Lem{detect} easily implies \Thm{main-theorem}. Once again, this is similar to the analysis in~\cite{KMS15}.

\def\EM{E_{\scriptscriptstyle{M}}}
\begin{lemma} \label{lem:detect} Suppose there exists a \textit{{\em matching} of violated edges} $\EM\subseteq S^{-}_f$ in $\AH$ of size $|\EM| = \sigma n^d$ (i.e. $\Gamma^-_f \geq \sigma$) for some $1\geq \sigma \geq \frac{\eps^{2/3}\log^{1/3}n }{d^{1/6}}$. Then the tester described in \Fig{alg} with inputs $f$ and $\eps$ detects a violation with probability $\Omega\Big(\frac{\sigma^2}{\sqrt{d}(\log^{3/2}d)(\log^2 n) }\Big) = \Omega\left(\frac{\eps^{4/3}}{d^{5/6}(\log^{3/2}d)(\log^{4/3} n)}\right)$.
\end{lemma}


\begin{proof} 

	Given $\sigma$, call $\tau$ "good" if $\frac{\sigma}{2 \log n} \sqrt{\frac{d}{10\log d}} \leq \tau \leq \frac{\sigma}{\log n} \sqrt{\frac{d}{10\log d}}$. Note that given the bounds on $\sigma$, the above range is a subinterval of $[1,\sqrt{d/10\log d}]$, and therefore the probability the sampled $\tau$ is good is $\Omega(1/\log d)$. Henceforth we condition on $\tau$ being good.
	
	Given $\tau$, \Lem{persistence} tells us that there are at most $O(\tau I_f/d)$ vertices~\footnote{We assume here that $\eps>1/d$ since otherwise the dependence on $d$ is not a meaningful quantity to study. In this case we have $I_f \geq \eps$ and thus $\tau I_f/d \gg d^{-9}$. We do not use the $\log n$ term in the denominator as this could in principle be $\gg d^9$. } which are $(\tau-1)$-non-persistent.
	For our choice of $\tau$, and since $I_f \leq O(\sqrt{d}\log n)$ by \Obs{inf}, we get that there are a $o(\sigma)$ fraction of the vertices which are $(\tau-1)$-non-persistent.
	
	Let $B = \{y | (x,y) \in \EM\}$; we have $|B| = \sigma n^d$ and let $B'\subseteq B$ be the $(\tau-1)$-persistent vertices of $B$. Let $A'$ be the $\EM$-matched vertices of $B'$. Note that $|A'| = |B'| = (1-o(1))\sigma n^d$ since the fraction of $(\tau-1)$-non-persistent vertices is $o(\sigma)$.
	Let $(x,z)$ be the pair sampled by the tester. Consider the following events.
	
	\begin{itemize}[noitemsep]
		\item $\calE_1$: $x$ lies in $A'$ with the edge $(x,y)\in E_M$ with $y\in B'$. Let $(x,y)$ lie in $H^{c_i}_{i,a_i}$.
		\item $\calE_2$: The matching sampled for dimension $i$ is indeed $H^{c_i}_{i,a_i}$.
		\item $\calE_3$: $i\in T$, that is, $i$ is one of the chosen $\tau$ dimensions.
		\item $\calE_4$: $f(y) = f(z)$, which means $f(z)=0$ and thus implies $(x,z)$ is a violation.
	\end{itemize}
	The probability the tester rejects is $\geq \Pr[\calE_1,\calE_2,\calE_3,\calE_4]$. Note that
	$\Pr[\calE_4|\calE_1,\calE_2,\calE_3] \geq 9/10$. This is because $\calE_1,\calE_2,\calE_3$ implies the distribution of $z$ on taking $\tau$ steps from $x$ is the same as taking $(\tau-1)$-steps from $y$. Since $y$ is $(\tau-1)$-persistent, we get the desired result. Also note
	$\calE_1,\calE_2, \calE_3$ are independent and therefore $\Pr[\calE_1,\calE_2,\calE_3] = \sigma(1-o(1))\cdot \Omega(1/\log n)\cdot (\tau/d)$.
	If $\tau$ is good, the latter is $\Omega\left(\frac{\sigma^2}{\sqrt{d}\log^{1/2} d\cdot \log^2 n}\right)$.
	Since the probability $\tau$ is good is $\Omega(1/\log d)$, the result follows.
\end{proof}


\noindent \textbf{Proof of \Thm{main-theorem}:} 
Let $\EM\subseteq S^{-}_f$ be the largest matching of violated edges, and let $|\EM| =\sigma n^d$. If $\sigma \geq \frac{\eps^{2/3}\log^{1/3}n}{d^{1/6}}$, then we are done by \Lem{detect}. Otherwise, by the Margulis-type \Thm{main-margulis}, we get $I^{-}_f = \Omega\left(\frac{\epsilon^{4/3} d^{1/6}}{\log^{1/3} n}\right)$. Since with probability $\Omega(1/\log d)$ the tester samples edges of the augmented hypergrid uniformly at random, and since  there are $O(n^d\cdot d\log n)$ edges in the augmented hypergrid,
the theorem follows.

\subsection{Proof of~\Lem{influence}: An Upper Bound on the Total Influence}\label{sec:influence}

In this section we prove~\Lem{influence}, a generalization of a theorem from \cite{KMS15} regarding the total influence of $f$. 
This fact may be of independent interest. Unlike the corresponding statement in~\cite{KMS15}, we weren't able to find a simpler proof than what we provide below. It uses the definition of Walsh functions described in~\cite{BlRY14} as a Fourier basis of functions over $[n]^d$ and the crux is proving the bound for the $d=1$ case. The $\sqrt{d}$ appears, as in the case of~\cite{KMS15}, from Parseval's identity and Cauchy-Schwarz.

We refer the reader to \Sec{fourier} for the facts we need about Walsh functions (these are from Blais et al~\cite{BlRY14}).
The Fourier coefficients are indexed by a collection of $d$-sets; for the rest of this section we need only one coefficient for each dimension $i$. In general, for $i \in [d]$ and $j \in [\log n]$, define $\be_{ij} := (S_1, ..., S_d)$  where $S_i = \{j\}$ and $S_{i'} = \emptyset$ for all $i' \neq i$.
In particular, the coefficient we need is $\widehat{f}(\be_{i,\log(n/2)})$ which evaluates to (see~\Obs{walsh})
%
%
\begin{align}\label{eq:def_walsh}
\widehat{f}(\be_{i,\log(n/2)}) = \frac{1}{2} \Expect_{(x,y)\in H^0_{i,\log(n/2)}} [f(x) - f(y)] 
\end{align}
That is, $\widehat{f}(e_{i,\log(n/2)})$ is the expected difference of $f$ evaluated on the endpoints of a u.a.r edge from the "longest" edge matching of dimension $i$, $H^0_{i,\log(n/2)}$~\footnote{Note that there is no corresponding "odd" matching for the step size $\log (n/2)$, say $H^1_{i,\log(n/2)}$, since all of the $2^{\log(n/2)} = n/2$ length edges can be contained in a single perfect matching.}. 
Henceforth, for notational simplicity we use $\be_i := \be_{i,\log (n/2)}$.
The only "Fourier Analysis" fact we need is $\sum_{i=1}^d \widehat{f}^2(\be_i) \leq 1$ (\Lem{Parseval}).

The proof follows by first arguing for the lines ($d=1$ case) and taking a direct sum. 
Recall $S_f$ is the set of sensitive edges partitioned into $S^+_f \cup S^-_f$. We have $I_f = |S_f|/n$,
and similarly $I^+_f$ and $I^-_f$ are defined. Let's define $\Delta I_f := I^+_f - I^-_f$.

\begin{lemma}\label{lem:influence_N_line} For any function $f:\AHl \to \{0,1\}$, $\Delta I_f \leq \log n\cdot\left(4I^-_f - \widehat{f}(\be_1)\right)$.
\end{lemma} \smallskip

\noindent
Before proving the above lemma, let us see how this implies~\Lem{influence}. 
Fix a dimension $i$ and let $L_i$ be the lines along dimension $i$. Fix an $\ell\in L_i$ and let $f_{|\ell}$ be the function restricted to $\ell$.
The above lemma gives
$-\widehat{f_{|\ell}(\be_i)} \geq \frac{\Delta I_{f_{|\ell}}}{\log n} - 4I^{-}_{f_{|\ell}}$.
Taking average over all lines $\ell \in L_i$ will set the LHS to $\widehat{f}(\be_i)$, while the RHS will give the contribution of $I^+_f$ and $I^-_f$ in the $i$th dimension. Summing over all dimensions gives
\begin{equation}\label{eq:pain}
\sum_{i=1}^d |\widehat{f}(\be_i)| \geq \frac{I^+_f - I^-_f}{\log n} - 4I^-_f = \frac{I_f - 2I^-_f}{\log n} - 4I^{-}_f > \frac{I_f}{\log n} - 6I^{-}_f
\end{equation}
where the $I^+_f,I^-_f$'s are defined over $\AH$ and the last inequality used $\log n > 1$. Since $I^-_f < \sqrt{d}$, we may assume $I_f > 7\log n\cdot I^{-}_f$, for otherwise~\Lem{influence} is trivially true. Thus the RHS of \Eqn{pain} $> \frac{I_f}{7\log n}$, while the LHS, by Cauchy-Schwarz and Parseval is at most $\sqrt{d}$. This proves~\Lem{influence}.

\begin{proof}[Proof of \Lem{influence_N_line}]
	Our first claim observes that the lemma is true for monotone $f$. 
\begin{claim}\label{clm:influence_M_line} If $f:\AHl \to \{0,1\}$ is monotone, then $\Delta I_f \leq \log n \cdot (-\widehat{f}(\be_1))$.
\end{claim}

\begin{proof} Consider the index $j \in [n]$ such that $f(i) = 0$ if $i \leq j$ and $f(i) = 1$ if $i > j$. We know such an index exists since $f$ is monotone. From \Eqn{def_walsh}, we get  $-\widehat{f}(\be_1) = \frac{\min (j, n-j)}{n}$ and $I_f^+ - I_f^- = I_f^+ \leq \frac{\min (j,n-j) \log n}{n}$ since the degree of a vertex in $\AHl$ is at most $\log n$. 
\end{proof}

Given $f:\AHl\to\{0,1\}$, let $S(f)$ be the "sorted" function which is the monotone Boolean function over $[n]$ with the same number of ones as $f$. The next claim shows that sorting only increases $\Delta I$.

\begin{claim}\label{clm:sorted} For any $f:\AHl\to \{0,1\}$, $\Delta I_{S(f)} \geq \Delta I_f$.
\end{claim}

\begin{proof} First, observe that the quantity $\Delta I_f$ restricted to a path is precisely equal to the difference of $f$ evaluated on the endpoints of that path. Bearing this in mind, we partition the edge set of $\AHl$ into a collection of paths. We say $\boldsymbol{p}$ is a $a$-path if it consists of edges from $\AHl$ only of length $2^a$ and $\boldsymbol{p}$ is maximal if no edges can be added to $\boldsymbol{p}$ to create a longer $a$-path. Observe that for $a$ in the range $0 \leq a < \log n$, there are $2^a$ maximal $a$-paths, each consisting of $2^{\log n - a}$ vertices and $2^{\log n - a} - 1$ edges. Let $\calP$ denote the set of paths $\boldsymbol{p}$ which are subgraphs of $\AHl$ and $\boldsymbol{p}$ is a maximal $a$-path for some $a \in [\log n]$. Let $\bp_s,\bp_t$ denote the start and end vertex of $\bp$, respectively. For $x \in [n]$ let $p^-(x)$,$p^+(x)$ denote number of paths $\boldsymbol{p} \in \calP$ for which $x$ is a start point or end point, respectively. It follows that 

$$\Delta I_f = \frac{1}{n} \sum_{\bp \in \calP} f(\bp_t) - f(\bp_s) = \frac{1}{n}\sum_{x \in [n]} f(x) \cdot (p^+(x) - p^-(x))$$

where the first equality is because $\AHl = \cup_{\boldsymbol{p} \in \calP}$, any two paths $\boldsymbol{p},\boldsymbol{p}' \in \calP$ are edge-disjoint and the observation made in the first line of this proof. The second equality is obtained simply by rewriting as a sum over the vertices.

Notice now that increasing $x$ can only increase $p^+(x)$ and decrease $p^-(x)$ since every path in $\calP$ is maximal. I.e., the functions $p^+(x)$ and $p^-(x)$ are monotone and anti-monotone, respectively. 
Thus, we can define $\boldsymbol{1}(f) = \{x \in [n] | f(x) = 1\}$ and $\boldsymbol{1}(S(f)) = \{x \in [n] | S(f)(x) = 1\}$ and observe that there is a bijection $\phi: \boldsymbol{1}(f) \to \boldsymbol{1}(S(f))$ such that $\forall x \in \boldsymbol{1}(f)$, $x < \phi(x)$. That is, sorting moves the 1's of $f$ to larger values on $\AHl$ and moves the 0's of $f$ to smaller values on $\AHl$ and since $p^+(x),p^-(x)$ are monotone and anti-monotone, the claim follows.
\end{proof}
The final claim connects the Walsh-coefficients.
\begin{claim}
	\label{clm:final}
	For any $f:\AHl\to\{0,1\}$, $-\widehat{S(f)}(\be_1) \leq -\widehat{f}(\be_1) + 4I^-_f$.
\end{claim}
\begin{proof}
	\def\dist{{\tt dist}}
	Suppose the distance between $f$ and $S(f)$ is $\delta$, that is, we can make $\delta n$ changes to $f$ to get $S(f)$.
	By definition of $\widehat{f}(\be_1)$, each change can increase $-\widehat{f}(\be_1)$ by at most additive $1/n$. Thus,
	$-\widehat{S(f)}(\be_1) \leq -\widehat{f}(\be_1) + \delta$. We next claim that $\delta$ is at most two times the distance, $\eps_f$ to monotonicity; this will prove the lemma since we know for the line $I^-_f \geq \eps_f/2$~\cite{EKK+00,ChSe13}.

 Consider the sorted version of $f$, $S(f) = 0^{j}1^{n-j}$ where $j = |\{x \in \ell | f(x) = 0\}|$. There are $z$ $1$'s in the $j$-prefix of $f$ and $z$ $0$'s in the $n-j$ suffix of $f$ where $\delta n = 2z$. Let $\Delta_p(f) = \{x \in [1,j] | f(x) = 1\}$ and $\Delta_s(f) = \{y \in [j+1,n] | f(y) = 0\}$. Changing the value of $f$ on any set of vertices $V$ where $|V| < z$ to get a new function $f'$ will result in $\Delta_p(f') \neq \emptyset$ and $\Delta_s(f') \neq \emptyset$. Thus, $f'$ cannot be monotone. That is, transforming $f$ into a monotone function requires that we change its value on at least $z$ vertices and so $\eps_f n \geq z = \delta n/2$. \end{proof}
\Lem{influence_N_line} follows from the previous three claims.
\end{proof}

%
\bibliography{derivative-testing}
\bibliographystyle{alpha}

\appendix

\section{Assuming $n$ is a power of $2$} \label{sec:power}

The reduction follows directly from the next theorem.

\begin{theorem} \label{thm:power} Given query access to $f:[n]^d \to \{0,1\}$,
we can simulate query access to $g:[N]^d \to \{0,1\}$ with the following properties.
\begin{asparaenum}
    \item $N$ is a power of $2$, and $N = \Theta(nd)$.
    \item If $f$ is monotone, so is $g$. If $f$ is $\eps$-far from monotone, then $g$ is $\eps/6$-far
    from monotone.
    \item A single query to $g$ can be simulated by a single query to $f$.
\end{asparaenum}
\end{theorem}

\begin{proof} First, we argue that there exists an integer $0 \leq i \leq d-1$ such that
$[n(d + i), n(d + i+1)]$ contains a power of $2$. This is because this set of disjoint intervals 
covers $[nd, 2nd]$ which contains a power of $2$. 
Fix $i$ to be the integer with the above property, and let $N \in [n(d + i), n(d + i+1)]$ be the power of $2$. Note $N = \Theta(nd)$.

Note that $N = n(d+i+1) - m$ for some non-negative integer $m\leq n$. It is convenient to rewrite this as
$N = m(d+i) + (n-m)(d+i+1)$. We define a map $\phi:[N] \to [n]$ as follows. If $y \leq m(d+i)$, set $\phi(y) = \lfloor \frac{y}{d+i} \rfloor$.
If $y > m(d+i)$, set $\phi(y) = m + \lfloor \frac{y-m(d+i)}{d+i+1} \rfloor$. It is instructive
to think of $\phi^{-1}$ which maps the first $m$ elements in $[n]$ to intervals in $[N]$ of size $d+i$,
and the remaining $n-m$ elements to intervals in $[N]$ of size $d+i+1$.
Abusing notation, we define $\phi:[N]^d \to [n]^d$ as 
$\phi(y_1, y_2, \ldots, y_d) = (\phi(y_1), \phi(y_2), \ldots, \phi(y_d))$.
Given any $x\in [n]^d$, note that $\phi^{-1}(x)$ is a sub-hypergrid in $[N]^d$ and for distinct $x,y\in [n]^d$
the sub-hypergrids $\phi^{-1}(x)$ and $\phi^{-1}(y)$ are disjoint. Finally the number of points in these sub-hypergrids
are between $(d+i)^d$ and $(d+i+1)^d$.

We define $g:[N]^d \to [n]^d$ as $g(y) = f(\phi(y))$. A single query
to $g$ can be answered using a single query to $f$. Since $\phi$ is monotone,
if $f$ is monotone, then so is $g$.

Suppose $f$ is $\eps$-far from monotone. Then there exists a matching 
$\cM$ of violated pairs in $[n]^d$ with $|\cM|\geq \eps n^d/2$. 
We now construct a matching in $[N]^d$ of size $\geq \eps N^d/6$
of pairs which violate monotonicity 
in $g$. 

Consider any pair $(x,y)$ in $\cM$ with $x \prec y$.
For each $j \in [d]$, let $\alpha_j$ and $\beta_j$ denote the number of length $d+i$ and length $d+i+1$ intervals, respectively, seperating $x_j$ from $y_j$.
Define the function $\psi_{(x,y)}:\phi^{-1}(x) \to \phi^{-1}(y)$ such that for $u \in \phi^{-1}(x)$, $\psi_{(x,y)}(u) := v$ where, for every $j \in [d]$, $v_j := u_j$ when $y_j = x_j$ and $v_j := u_j + \alpha_j(d+i) + \beta_j(d+i+1)$ when $y_j \neq x_j$. Thus, $\psi_{(x,y)}$ maps $u$ to $v$ which for each $j \in [d]$ is $y_j - x_j$ many sub-grids away from $u$ in dimension $j$. Observe that $\psi_{(x,y)}$ is injective (not necessarily bijective since possibly $|\phi(x)^{-1}| < |\phi(y)^{-1}|$) and for any $u \in \phi^{-1}(x)$: (a) $u \prec \psi_{(x,y)}(u)$ and (b) $g(u) > g(\psi_{(x,y)}(u))$.

Now, for $(x,y) \in \cM$, define $M_{(x,y)} := \{(u,\psi_{(x,y)}(u)) : u \in \phi_{-1}(x)\}$ and observe that this is a matching of violated pairs in $[N]^d$ with respect to $g$ and $|M_{(x,y)}| \geq (d+i)^d$. Moreover, for two different pairs $(x,y),(x',y') \in \cM$, $M_{(x,y)}$ and $M_{(x',y')}$ are disjoint and their union forms a matching.

Thus, taking the union over all pairs in $\cM$,
there is a violation matching (for $g$) of size at least 
\[
\frac{\eps n^d}{2} \cdot (d+i)^d \geq \frac{\eps}{2} \cdot [n(d+i+1)]^d \cdot \frac{(d+i)^d}{(d+i+1)^d} \geq \frac{\eps N^d}{6}
\]
and so $g$ is $\geq \eps / 6$-far from monotone. \end{proof}

\section{Fourier Analysis on $\AH$}\label{sec:fourier}

The following definitions and facts are due to Blais et al \cite{BlRY14}. Note that this section relies heavily on $n$ being a power of $2$. Indeed this is assumed in \cite{BlRY14}. Refer to \Sec{power} for our reduction to this case and to \Sec{discussion_power} for our discussion on this point.

\begin{definition} [$1$-dimensional Walsh functions] \label{def:walsh} For $i \in [\log n]$, define $w_i: [n] \to \{-1,1\}$ by $w_i(x) := (-1)^{bit_i(x-1)}$. For any $S \subseteq [\log n]$, the \textit{Walsh function} $w_S: [n] \to \{-1,1\}$ corresponding to $S$ is $$w_S(x) := \prod_{i\in S} w_i(x) \text{.}$$ If $S = \emptyset$, then $w_S(x) := 1$. 
\end{definition}

\begin{definition} [$d$-dimensional Walsh functions] \label{def:walsh_d} The \textit{d-dimensional Walsh function} $w_{\mathbf{S}}: [n]^d \to \{-1,1\}$ corresponding to the vector $\mathbf{S} = (S_1, ..., S_d)$ of subsets $S_i \subseteq [\log n]$ is defined by $$w_{\mathbf{S}}(x_1, ...,x_d) := \prod_{i=1}^d w_{S_i}(x_i) \text{.}$$
\end{definition}

The set of $d$-dimensional Walsh functions form a fourier basis for the set of functions over $[n]^d$. That is, every function $f: [n]^d \to \{-1,1\}$ can be expressed as $$f(x) = \sum_{\substack{\mathbf{S} = (S_1, ..., S_d) \\ S_i \subseteq [\log n]}} \widehat{f}(\mathbf{S}) w_{\mathbf{S}}(x)$$ where $\widehat{f}(\mathbf{S})$ is the fourier coefficient of $f$ corresponding to $\mathbf{S}$. The fourier coefficients both contain useful information about $f$ and satisfy some convenient properties. The following two facts are proved by \cite{BlRY14}.

\begin{fact}\label{fact:1} For any $\mathbf{S},\mathbf{T}$, the $d$-dimensional Walsh function $w_{\mathbf{S} \Delta \mathbf{T}}:[n]^d \to \{-1,1\}$ corresponding to the symmetric difference of $\mathbf{S}$ and $\mathbf{T}$ satisfies $w_{\mathbf{S} \Delta \mathbf{T}} = w_{\mathbf{S}} \cdot w_{\mathbf{T}}$. 
\end{fact}

\begin{fact}\label{fact:2} $\Expect_{x \in [n]^d}[w_{\mathbf{S}}(x)] = 0$ unless $\mathbf{S} = \emptyset^d$, in which case $\Expect_{x \in [n]^d}[w_{\mathbf{S}}(x)] = 1$. \\
\end{fact}

The following lemma gives a useful interpretation of the coefficient $\widehat{f}(\mathbf{S})$ and follows from the above facts.

\begin{lemma}\label{lem:fourier_coeff} $\widehat{f}(\mathbf{S}) = \Expect_{x \in [n]^d} [f(x)w_{\mathbf{S}}(x)]$. 
\end{lemma}

\begin{proof} Using the fourier expansion of $f$ and applying \Fact{1} we get $$f(x)w_{\mathbf{S}}(x) = \sum_{\mathbf{T}} \widehat{f}(\mathbf{T}) w_{\mathbf{T}}(x) \cdot w_{\mathbf{S}}(x) = \sum_{\mathbf{T}} \widehat{f}(\mathbf{T}) w_{\mathbf{S} \Delta \mathbf{T}}(x) \text{.}$$ Applying \Fact{2}, taking the expectation of each side and applying linearity yields $$\Expect_{x \in [n]^d} [f(x)w_{\mathbf{S}}(x)] = \sum_{\mathbf{T}} \widehat{f}(\mathbf{T}) \Expect_{x \in [n]^d} [w_{\mathbf{S} \Delta \mathbf{T}}(x)] = \widehat{f}(\mathbf{S})w_{\emptyset^d} = \widehat{f}(\mathbf{S}) \text{.}$$ \end{proof}

\begin{lemma} [Parseval's for Hypergrids] \label{lem:Parseval} For any function $f:[n]^d \to \{-1,1\}$, $\sum_{\mathbf{S}} \widehat{f}(\mathbf{S})^2 = 1$.
\end{lemma}

\begin{proof} 
\begin{align*}
    1 = \Expect_{x} [f(x)^2] &= \Expect_{x} \Big[\big(\sum_{\mathbf{S}} \widehat{f}(\mathbf{S}) w_{\mathbf{S}}(x)\big)^2\Big] \nonumber = \Expect_{x} \Big[ \sum_{\mathbf{S},\mathbf{T}} \widehat{f}(\mathbf{S})\widehat{f}(\mathbf{T}) w_{\mathbf{S}}(x)w_{\mathbf{T}}(x) \Big] \nonumber \\
      &= \sum_{\mathbf{S},\mathbf{T}} \widehat{f}(\mathbf{S})\widehat{f}(\mathbf{T}) \Expect_x[w_{\mathbf{S} \Delta \mathbf{T}}(x)] = \sum_{\mathbf{S}} \widehat{f}(\mathbf{S})^2
\end{align*}

where the last equality is due to \Fact{2}. \end{proof}

For $i \in [d]$ and $j \in [\log n]$, define $\be_{ij} := (S_1, ..., S_d)$  where $S_i = \{j\}$ and $S_{i'} = \emptyset$ for all $i' \neq i$. By the previous lemma we have $\hat{f}(\be_{ij}) = \Expect_{x} [f(x) w_{\be_{ij}}(x)]$ and by definition we see that 

\[ w_{\be_{ij}}(x) = 
    \begin{cases} 
      1 & \text{ if } x_i \pmod{2^{j+1}} < 2^j \\
      -1  & \text{ if } x_i \pmod{2^{j+1}} \geq 2^j 
   \end{cases} \text{.}
\]

That is, $w_{\be_{ij}}(x)$ is 
$1$ when $x$ is a \emph{lower} endpoint in the matching $H_{i,j}^0$ and $-1$ when $x$ is an \emph{upper} endpoint in $H_{i,j}^0$ (recall \Sec{augmented_hypergrid}). 
Let $I^+_{f,i,j}$, $I^-_{f,i,j}$ denote the positive and negative influence, respectively, of $f$ restricted to the edges in $H_{i,j}^0$. Combining this with \Lem{fourier_coeff} we make the following useful observation which we use to prove \Lem{influence}.

\begin{observation}\label{obs:walsh}$\widehat{f}(\be_{ij}) = \frac{1}{2} \Expect_{(x,y)\in H^0_{i,j}} [f(x) - f(y)]$. \end{observation}



 


\section{Proof of \Thm{edge}}\label{sec:edge}

    In this section we prove a lower bound on the number of violating edges in $\AH$ using alternating paths machinery developed by Chakrabarty and Seshadhri in \cite{ChSe13-j,ChSe13}. 

\subsection{Structure of the Edge Matchings}




Recall from \Sec{augmented_hypergrid} that $\boldsymbol{H} := \{H_{i,a}^c | i \in [d], a \in [\log n], c \in \{0,1\} \}$ is the collection of edge matchings in $\AH$. For a matching $H \in \boldsymbol{H}$, let $L(H) := \{x|\exists (x,y)\in H\}$ and $U(H):= \{y|\exists (x,y)\in H\}$ denote the lower and upper endpoints, respectively. For each $H$, partition $M$ into three sets: $(x,y)$'s that use a $H$-edge on any shortest path from $x$ to $y$ ($H$-cross pairs), $(x,y)$'s that do not use a $H$-edge on any shortest path from $x$ to $y$ \textit{and} have both $x$ and $y$ in the same $L(H)$ or $U(H)$ ($H$-straight pairs) and $(x,y)$'s that do not use a $H$-edge on any shortest path from $x$ to $y$ \textit{and} have one endpoint in $L(H)$ and one in $U(H)$ ($H$-skew pairs). That is, \\

\noindent - $cr_{H_{i,a}^c}(M)$: $(x,y) \in M$ such that $x$ to $y$ shortest paths contain an edge from $H$ and $x \in L(H_{i,a}^c)$. 

\noindent - $st_{H_{i,a}^c}(M)$: $(x,y) \in M$ such that $x$ to $y$ shortest paths \textit{do not} contain an edge from $H$ and $x,y \in L(H_{i,a}^c)$ or $x,y \in U(H_{i,a}^c)$.

\noindent - $sk_{H_{i,a}^c}(M)$: $(x,y) \in M$ such that $x$ to $y$ shortest paths \textit{do not} contain an edge from $H$ and $x \in L(H_{i,a}^c)$, $y \in U(H_{i,a}^c)$ or $x \in U(H_{i,a}^c)$, $y \in L(H_{i,a}^c)$. \\

When $(x,y) \in cr_{H}(M)$, we say $(x,y)$ \textit{crosses} $H$. 

\subsection{The Potential Function}

For a pair $(x,y)$ and a matching $H := H_{i,a}^c \in \boldsymbol{H}$, define 

$$\mu_H(x,y) = 
\begin{cases} 
    \frac{1}{2^a} \text{ if }  x,y \in L(H) \text{ or } x,y \in U(H) \\
    0 \text{ if }  x \in L(H), y \in U(H) \text{ or } y \in L(H), x \in U(H) \text{.}
\end{cases}
$$

Note that $\mu_H(x,y)$ is $0$ for $H$-skew and $H$-crossing pairs while it is positive only for $H$-straight pairs. Importantly, this positive value is larger for $H$ with small step size and in particular $1/2^a > \sum_{a'>a} 1/2^{a'}$. Thus, the following potential is designed to correct $H$-skew pairs by aligning endpoints with respect to $H_{i,a}$ for which $a$ is small:

$$\Phi(M) = \sum_{(x,y)\in M} \sum_{H \in \boldsymbol{H}} \mu_{H}(x,y) \text{.}$$

Note that $\Phi(\cdot)$ has the same effect as the potential function described in \cite{ChSe13}.

\subsection{Main Proof}

The following lemma establishes that every $H$-cross pair implies the existence of a unique violating $H$-edge. 

\begin{lemma}\label{lem:violations}
    $H_{i,a}^c$ contains at least $|cr_{H_{i,a}^c}(M)|/2$ violations. 
\end{lemma}

\Lem{violations} is the main tool for proving \Thm{edge}. We defer the proof to \Sec{violations_proof} and proceed with the proof of \Thm{edge}. \\

\noindent \textbf{Proof of \Thm{edge}:} Let $M$ be a maximal matching in the violation graph of $f$ with the \textit{smallest} average length, $$r = |M|^{-1} \sum_{(x,y) \in M} d_{\AH}(x,y)$$ and among such matchings, be one that maximizes $\Phi(M)$. Notice that each $(x,y) \in M$ crosses exactly $d_{\AH}(x,y)$ matchings $H \in \boldsymbol{H}$. Moreover, $H$ contains at least $|cr_H(M)|/2$ violations by \Lem{violations}. Let $\Delta f$ denote the total number of violating edges in $\AH$ and let $\Delta f |_H$ denote the number of violating $H$-edges. We have

$$\Delta f = \sum_{H \in \boldsymbol{H}} \Delta f|_H \geq \frac{1}{2}\sum_{H \in \boldsymbol{H}} |cr_{H}(M)| = \frac{1}{2} \sum_{(x,y)\in M} d_{\AH}(x,y) = \frac{r|M|}{2} \geq \frac{r\epsilon n^d}{4} \text{.}$$ 

\subsection{Proof of \Lem{violations}}\label{sec:violations_proof}

Let $H$ denote an arbitrary $H_{i,a}^c$-matching and let $X := \{x : (x,y) \in cr_{H}(M)\}$ be the set of start points of $cr_{H}(M)$. For $x \in X$, define the sequence $S_x$ as follows:

\begin{itemize} [noitemsep]
	\item $S_x(0) = x$.
	\item $S_x(j+1) = H(S_x(j))$ if $j$ is even. Terminate at $j+1$ if $S_x(j)$ and $H(S_x(j))$ are a violation to monotonicity.
	\item $S_x(j+1) = M(S_x(j))$ if $j$ is odd. Terminate if $S_x(j)$ is $st_H(M)$-unmatched. 
\end{itemize}

We will show that every $S_x$ contains a violating edge in $H$ and there are at least $|X|/2$ disjoint sequences. For notational simplicity, let $s_j := S_x(j)$ where $S_x(j)$ is defined and let $s_{-1} := y$. For simplicity, we will also use $d$ in place of $d_{\AH}$. 

Note that $S_x$ alternates between the matchings $H$ and $st_H(M)$. In particular, $S_x$ never follows a skew-pair. Thus the structure of $S_x$ on $\AH$ is in many senses the same as the structure of $S_x$ on the hypercube $\{0,1\}^d$, where there are no skew-pairs. Thus, we get the following two claims regarding the structure of $S_x$ which have proofs identical to the proofs of Claims 2.9.1. and 2.9.2 in \cite{ChSe13-j}. 

\begin{claim}\label{clm:S_struc_1}
If $s_j$ is defined, then the following holds:
    \begin{itemize} [noitemsep]
        \item $j \equiv 0 \pmod{4}$ $\Longleftrightarrow$ $f(s_j) = 1$ and $s_j \in L(H)$. 
        \item $j \equiv 1 \pmod{4}$ $\Longleftrightarrow$ $f(s_j) = 1$ and $s_j \in U(H)$. 
        \item $j \equiv 2 \pmod{4}$ $\Longleftrightarrow$ $f(s_j) = 0$ and $s_j \in U(H)$. 
        \item $j \equiv 3 \pmod{4}$ $\Longleftrightarrow$ $f(s_j) = 0$ and $s_j \in L(H)$. 
    \end{itemize}
\end{claim}

\begin{proof} The proof is exactly analogous to the proof of Claim 2.9.1. in \cite{ChSe13-j}. \end{proof}

\begin{claim}\label{clm:S_struc_2}
If $s_j$ is defined and $j$ \textit{is odd}, then 

\begin{itemize} [noitemsep]
    \item $j \equiv 1 \pmod{4}$ $\Longleftrightarrow$ $s_{j-3} \succ s_j$ and $d(s_j,s_{j-3}) = d(s_{j-1},s_{j-2})$.
    \item $j \equiv 3 \pmod{4}$ $\Longleftrightarrow$ $s_{j-3} \prec s_j$ and $d(s_{j-3},s_j) = d(s_{j-2},s_{j-1})$.
\end{itemize}
\end{claim}

\begin{proof} The proof is exactly analogous to the proof of Claim 2.9.2. in \cite{ChSe13-j}. \end{proof}

\begin{lemma}\label{lem:S_violation}
    For any $(x,y) \in cr_H(M)$, $S_x$ contains a violating edge in $H$.
\end{lemma}

\begin{proof} Suppose for the sake of contradiction that $S_x$ terminates at $s_j$ ($j$ is odd) without witnessing a violation in $H$. Define the two matchings 

$$E_{-}(j) := \{(s_0,s_{-1}),(s_1,s_2), (s_3,s_4) ..., (s_{j-2},s_{j-1}) \} \subseteq M$$

and 

$$E_{+}(j) := \{(s_1,s_{-1}), (s_0,s_3), (s_2,s_5) , ..., (s_{j},s_{j-3}) \} \text{.}$$

Notice $E_{-}(j)$ involves $s_{-1},s_0, ..., s_{j-2}, s_{j-1}$ while $E_{+}(j)$ involves $s_{-1},s_0, ..., s_{j-2}, s_j$. Clearly, $E_-(j)$ is a collection of pairs from $M$. Notice $|E_-(j)| = |E_+(j)|$ and let $d(E)$ denote the average distance of $E \in \{E_-(j), E_+(j)\}$. The following claim about $E_+(j)$ will be crucial for the rest of the proof.

\begin{claim}\label{clm:struc_E}
    $E_+(j)$ is a set of violating pairs with (a) $d(E_+(j)) = d(E_-(j)) -1$ and (b) $\Phi(E_+(j)) > \Phi(E_-(j))$.
\end{claim}

\begin{proof} \textbf{Proof of (a)}: For all \textit{odd} $k$ in the range $3 \leq k \leq j$ we have the following by combining claims \ref{clm:S_struc_1} and \ref{clm:S_struc_2}. $k \equiv 1 \pmod{4}$ implies $d(s_k,s_{k-3}) = d(s_{k-1},s_{k-2})$ and $f(s_k) = 1$, $f(s_{k-3}) = 0$. $k \equiv 3 \pmod{4}$ implies $d(s_{k-3},s_k) = d(s_{k-2},s_{k-1})$ and $f(s_k) = 0$, $f(s_{k-3}) = 1$. Finally, $d(s_{-1},s_1) = d(s_0,s_{-1}) - 1$ since the shortest path joining $s_0$ and $s_{-1}$ contains a $H$-edge by definition of $cr_H(M)$ and $s_0$ differs from $s_1$ only along $H$. \\

\noindent \textbf{Proof of (b)}: Consider an arbitrary \textit{odd} $k$ in the range $3 \leq k \leq j$. Suppose $j \equiv 3 \pmod{4}$ ($j \equiv 1$ is analogous). We have $(s_{k-3},s_k) \in E_+(j)$ and $(s_{k-2},s_{k-1}) \in E_-(j)$. Also, $H(s_{k-3}) = s_{k-2}$, $s_{k-3} \prec s_{k-2}$ and $H(s_k) = s_{k-1}$, $s_k \prec s_{k-1}$ by \Clm{S_struc_2}. Thus, $s_{k-3},s_k \in L(H)$ and $s_{k-2},s_{k-1} \in U(H)$ so $\mu_H((s_{k-3},s_k)) = \mu_H((s_{k-2},s_{k-1})) = 0$ and clearly $\mu_{H'}((s_{k-3},s_k)) = \mu_{H'}((s_{k-2},s_{k-1}))$ for any other $H' \in \boldsymbol{H}$. Finally, $\mu_H((s_{0},s_{-1})) > 0$ since $(s_{0},s_{-1}) \in cr_H(M)$ while $\mu_H((s_1,s_{-1})) = 0$ since $(s_1,s_{-1}) \in st_H(M)$. Also, $\mu_{H'}((s_0,s_{-1})) = \mu_{H'}((s_{1},s_{-1}))$ for any other $H' \in \boldsymbol{H}$. Thus $\Phi(E_+(j)) > \Phi(E_-(j))$. \end{proof}

\noindent \textit{Case (1):} $s_j$ is $M$-unmatched. In this case we can rearrange $M$ to get a new matching $M'$ with the same cardinality as $M$ and strictly smaller average distance. Let

$$M' := M \setminus E_{-}(j) \cup E_{+}(j) \text{.}$$

\Clm{struc_E} shows that $M'$ is a matching with the same cardinality as $M$ and average distance strictly less than that of $M$. Contradiction. \\

\noindent \textit{Case (2):} $s_j$ is $cr_H(M)$-matched. In this case we can again rearrange $M$ to get a new matching $M'$ with the same cardinality as $M$ and strictly smaller average distance. Let

\begin{align}
    M' := M \setminus \big(E_{-}(j) \cup \{ (s_j,M(s_j)) \}\big) \cup \big(E_{+}(j) \cup \{ (s_{j-1},M(s_j)) \}\big) \text{.} \label{eq:matching}
\end{align}

By \Clm{S_struc_1} and the fact that $j$ is odd, it follows that $f(s_j) = f(s_{j-1})$. Suppose $j \equiv 1 \pmod{4}$ (the other case is symmetrical). Then $s_{j-1} \prec s_j \prec M(s_j)$ and by definition of $cr_H(M)$, any shortest path from $s_j$ to $M(s_j)$ contains an $H$-edge. Since $s_{j-1}$ and $s_j$ differ exactly on a single $H$-edge, it follows that $d(s_{j-1},M(s_j)) = d(s_j,M(s_j))$ since we simply are replacing an $H$-edge (say of length $2^a$) by an edge of twice its length (of length $2^{a+1}$). Combining this fact with \Clm{struc_E} shows $M'$ again decreases the average distance and $|M'| = |M|$. Contradiction. \\ 

\noindent \textit{Case (3):} $s_j$ is $sk_H(M)$-matched. In this case we rearrange to get $M'$ with the same number of pairs, \textit{at most} the same average distance and $\Phi(M') > \Phi(M)$. Define $M'$ as in line \Eqn{matching}. Suppose $j \equiv 1 \pmod{4}$ (the other case is symmetrical). Again, $f(s_j) = f(s_{j-1})$ by \Clm{S_struc_1} and $s_{j-1} \prec s_j \prec M(s_j)$. However, since $(s_j, M(s_j)) \in sk_H(M)$ we cannot say that the distance stays the same. However, we do know that the distance can increase only by at most $1$, i.e. $d(s_{j-1},M(s_j)) \leq d(s_j,M(s_j)) + 1$ since $s_{j-1}$ and $s_j$ differ only by a $H$-edge. Thus, the average distance of $M'$ is at most the average distance of $M$, but as shown in the following claim, $\Phi$ increases. Contradiction.

\begin{claim}
    $\Phi(M') > \Phi(M)$.
\end{claim}

\begin{proof} Since $\Phi(E_+(j)) > \Phi(E_-(j))$ by \Clm{struc_E} it suffices to show that $\Phi(\{(H(s_j),M(s_j))\}) \geq \Phi(\{(s_j,M(s_j))\})$. First, observe that $\mu_H(s_j,M(S_j)) = 0$ (since this is a $H$-skew pair) and so $\mu_H(H(s_j),M(s_j)) = 1$. Suppose $H = H_{i,a}^c$ and let $H' = H_{i',a'}^{c'}$ be some other matching. Observe that $\mu_{H'}(H(s_j),M(s_j)) = \mu_{H'}(s_j,M(s_j))$ whenever $i' \neq i$ since $s_j$ and $H_j$ only differ on dimension $i$. In the case that $i' = i$ we can show that $\mu_{H'}(H(s_j),M(s_j)) = \mu_{H'}(s_j,M(s_j))$ whenever $a' < a$. This will strictly increase $\Phi(\cdot)$ since $1/2^a > \sum_{a' > a} 1/2^{a'}$ and so any affect on $\Phi(\cdot)$ exerted by matchings with step size $a' > a$ in dimension $i$ will be dominated by the fact that the pair $(H^c_{a,b}(s_j),M(s_j))$ is $H^c_{a,b}$-straight, while $(s_j,M(s_j))$ is not.

Suppose $a' < a$. Notice that for $x \in [n]^d$, the value $x_i \pmod{2^{a+1}}$ determines whether $x \in L(H^c_{i,a})$ or $x \in U(H^c_{i,a})$. E.g. $x \in L(H^0_{i,a})$ when $x < 2^a \pmod{2^{a+1}}$ and $x \in U(H^0_{i,a})$ when $x \geq 2^a \pmod{2^{a+1}}$. Thus $H(s_j)^{(i)} \equiv s_{j}^{(i)} \pm 2^a \equiv s_{j}^{(i)} \pmod{2^{a'+1}}$ where $H(s_j)^{(i)}$ and $s_{j}^{(i)}$ denote the $i$'th coordinate of $H(s_j)$ and $s_j$, respectively, and this implies $H(s_j)$ lies in the same end of the matching $H'$ as $s_j$. That is, $H(s_j) \in L(H')$ if and only if $s_j \in L(H')$. \end{proof}

All three cases imply a contradiction and so the only way for $S_x$ to terminate is by witnessing a violating edge in $H$. Since $S_x$ must terminate at some point, this proves \Lem{S_violation}. \end{proof}

Finally, $S_x$ and $S_y$ are disjoint \textit{unless} $y$ terminates $S_x$. Thus, there are at least $|X|/2 = |cr_H(M)|/2$ disjoint sequences, each containing a violation in $H$. This completes the proof of \Lem{violations}.

\end{document}